\documentclass{ws-ijmpa}
\usepackage[super,compress]{cite}
\usepackage{graphicx}
\pdfoutput=1
\usepackage{float}
\usepackage{latexsym}
\usepackage{amsmath,amstext,amssymb,amsfonts,
amscd,bm,array,multirow,amsbsy,mathrsfs,calc}
\usepackage{t1enc}
\usepackage{indentfirst}
\usepackage{pb-diagram}
\usepackage{graphicx}
\usepackage{enumerate}
\usepackage{color}
\usepackage[all]{xy}
\usepackage{hyperref}
\hypersetup{colorlinks=false,pdfborderstyle={/S/U/W 0}}
\usepackage[usenames,dvipsnames]{xcolor}
\usepackage{url}
\usepackage{upgreek}

\newtheorem{thm}{Theorem}[section]

\newtheorem{cor}[thm]{Corollary}

\newcommand{\be}{\begin{equation*}}
\newcommand{\ee}{\end{equation*}}
\newcommand{\ben}{\begin{equation}}
\newcommand{\een}{\end{equation}}
\newcommand{\beqa}{\begin{eqnarray*}}
\newcommand{\eeqa}{\end{eqnarray*}}
\newcommand{\beqan}{\begin{eqnarray}}
\newcommand{\eeqan}{\end{eqnarray}}
\newcommand{\nn}{\nonumber}

\def\Z{\mathbb{Z}}

\def\R{\mathbb{R}}

\def\Hess{\mathrm{Hess}}

\def\Crit{\mathrm{Crit}}

\def\pd{\partial}

\def\reg{\mathrm{reg}}
\def\sing{\mathrm{sing}}
\def\crit{\mathrm{crit}}
\def\noncrit{\mathrm{noncrit}}
\def\infl{\mathrm{infl}}

\def\rd{\mathrm{d}}

\def\fM{\mathfrak{M}}

\def\dd{\mathrm{d}}

\def\bepsilon{\boldsymbol{\varepsilon}}
\def\veta{\boldsymbol{\eta}}
\def\hveta{{\hat \veta}}
\def\bOmega{\boldsymbol{\Omega}}

\newcommand{\Tr}{\mathrm{Tr}}

\def\cF{\mathcal{F}}
\def\cG{\mathcal{G}}
\def\cH{\mathcal{H}}
\def\cL{\mathcal{L}}
\def\cM{\mathcal{M}}
\def\cN{\mathcal{N}}

\def\cS{\mathcal{S}}
\def\cX{\mathcal{X}}

\def\heta{{\hat \eta}}
\def\hetap{\heta^\parallel}
\def\hepsilon{{\hat \bepsilon}}

\def\rS{\mathrm{S}}

\def\fM{\mathfrak{M}}
\def\fS{\mathfrak{S}}
\def\fs{\mathfrak{s}}

\def\f{{\bf f}}

\def\q{{\bf q}}

\def\kin{\mathrm{kin}}
\def\pot{\mathrm{pot}}

\newcommand{\eqdef}{\stackrel{{\rm def.}}{=}}

\def\red{\mathrm{red}}

\def\grad{\mathrm{grad}}

\def\vol{\mathrm{vol}}

\def\cR{\mathcal{R}}

\begin{document}

\markboth{C. I. Lazaroiu}{Natural observables and dynamical approximations in multifield cosmological models}


\title{Natural observables and dynamical approximations in multifield cosmological models}

\author{C. I. Lazaroiu}

\address{Departamento de Matematicas, Universidad UNED - Madrid\\
 Calle de Juan del Rosal 10, 28040, Madrid, Spain\\
clazaroiu@mat.uned.es\\
and\\
Horia Hulubei National
  Institute of Physics and Nuclear Engineering,\\
  Reactorului 30, Bucharest-Magurele, 077125, Romania\\
lcalin@theory.nipne.ro}

\maketitle


\begin{abstract}
I give a geometric construction of certain first order natural
dynamical observables in multifield cosmological models with arbitrary
target space topology and discuss a system of related dynamical
approximations and regimes for such models.

\end{abstract}

\section{Introduction}

Cosmological models with more that one scalar field (``inflaton'') are
natural and preferred in quantum theories of gravity\cite{AP} and have
acquired increasing importance in recent years, being subject to
renewed study from various perspectives\cite{Lilia1,Lilia2,LiliaRev,
  Paban, AL, Noether1, Noether2, Hesse}.

Despite their importance in connecting cosmology to supergravity and
string theory, such models are poorly understood when compared to
their one field counterparts -- in particular because their dynamics
can be very involved already before considering
perturbations. Moreover, the scalar fields of such models can be
valued in an arbitrary connected Riemannian manifold (called {\em
  scalar manifold}) which is generally non-compact. Hence a systematic
study must consider scalar manifolds of arbitrary connected topology.

When the Hubble parameter is positive, the background dynamics of
multifield models can be described mathematically\cite{ren} using a
geometric dynamical system\cite{Palis, Katok} which is defined on the
tangent bundle to the scalar manifold. The general behavior of this
system can be extremely non-trivial already in the case of models with
two scalar fields \cite{genalpha, elem, modular, unif, Nis, Tim19}.

In general, one cannot hope to solve multifield cosmological dynamics
exactly or even to approach it numerically in a satisfactory manner
when the number of scalar fields (i.e. the dimension of the scalar
manifold) is large -- and even less so when the topology of the scalar
manifold is involved.  Hence the best hope for practical progress is
to develop approximation schemes which might allow one to make
progress -- such as the IR and UV expansions\cite{ren,grad}. To do
this in general, one needs a geometric description of the cosmological
observables which control various approximations and regimes of
physical interest while taking into account the topology of the scalar
manifold. In this paper, I give such a description for a collection of
natural cosmological observables and discuss some aspects of the
approximation schemes which they suggest.

The paper is organized as follows. Section \ref{sec:models} recalls
the description of general multifield cosmological models and gives a
geometric treatment of their scalar and vector observables. Section
\ref{sec:obs} gives a careful geometric construction of certain first
order observables (some of which are obtained by ``on-shell''
reduction of second order observables) for scalar manifolds of
arbitrary topology, discusses some relations between them and describes
various regions of interest which they define inside the tangent
bundle of the scalar manifold. Section \ref{sec:SR} discusses two
dynamical approximations which are suggested by the first slow roll
conditions. Section \ref{sec:CD} discusses the conservative and
dissipative approximations, which are controlled by the {\em
  conservative function} of the model -- one of the basic natural
observables introduced in Section \ref{sec:obs}. It also gives an
overview of other natural approximations and their relations to the
conservative and dissipative regimes. Section \ref{sec:Planck} briefly
describes the limits of large and small Planck mass, relating them to
the conservative and IR approximations. Section \ref{sec:conclusions}
presents our conclusions and a few directions for further research.

\paragraph{Notations and conventions.}

All manifolds considered in the paper are smooth and paracompact.  For
any manifold $\cM$, we let ${\dot T}\cM$ denote the slit tangent
bundle of $\cM$, defined as the complement of the image of the zero
section in $T\cM$. We denote by $\pi:T\cM\rightarrow \cM$ the bundle
projection and by $F:=F\cM\eqdef \pi^\ast(T\cM)$ the {\em Finsler
  bundle}\cite{SLK} of $\cM$, which is a bundle over $T\cM$. A section
of $F$ defined over an open subset of $T\cM$ is called a {\em Finsler
  vector field}. We denote by $\dot{F}\eqdef \pi^\ast(\dot{T}\cM)$ the
slit Finsler bundle of $T\cM$. A global section of $F$ is the same as
a map $f:T\cM\rightarrow T\cM$ which satisfies $f(u)\in T_{\pi(u)}\cM$
for all $u\in T\cM$, while a global section of ${\dot F}$ is a map of
this type which also satisfies $f(u)\neq 0$ for all $u$.

\section{Multifield cosmological models}
\label{sec:models}

In this section we recall the geometric description of multifield
cosmological models with arbitrary target space topology (in the
absence of cosmological perturbations) and give a geometric treatment
of their local observables. We also discuss some natural constructions
afforded by the Finsler bundle of the scalar manifold of such models,
which will be useful in later sections.

\subsection{Basics}

\noindent Throughout this paper, a {\em multifield cosmological model}
is a classical cosmological model with $d$ scalar fields (where $d\in
\Z_{>0}$ is a fixed positive integer) which is derived from the
following action on a spacetime with topology $\R^4$:
\ben
\label{S}
S[g,\varphi]=\int \vol_g \cL[g,\varphi]~~,
\een
where:
\ben
\label{cL}
\cL[g,\varphi]=\frac{M^2}{2} \mathrm{R}(g)-\frac{1}{2}\Tr_g \varphi^\ast(\cG)-\Phi\circ \varphi~~.
\een
Here $M$ is the reduced Planck mass, $g$ is the spacetime metric on
$\R^4$ (taken to be of ``mostly plus'' signature) while $\vol_g$ and
$\mathrm{R}(g)$ are the volume form and Ricci scalar of $g$. The
scalar fields are described by a smooth map $\varphi:\R^4\rightarrow
\cM$, where $\cM$ is a (generally non-compact) connected, smooth and
paracompact manifold of dimension $d$ which is endowed with a smooth
Riemannian metric $\cG$, while $\Phi:\cM\rightarrow \R$ is a smooth
function which plays the role of potential for the scalar fields. We
require that $\cG$ is complete to ensure conservation of energy. For
simplicity, we also assume that $\Phi$ is strictly positive on
$\cM$. The ordered system $(\cM,\cG,\Phi)$ is called the {\em scalar triple}
of the model while the Riemannian manifold $(\cM,\cG)$ is called the {\em scalar manifold}.
The model is parameterized by the quadruplet:
\ben
\fM=(M_0,\cM,\cG,\Phi)~~,
\een
where:
\be
M_0\eqdef M\sqrt{\frac{2}{3}}
\ee
is the {\em rescaled Planck mass}. We denote by:
\be
\Crit\Phi\eqdef \{m\in \cM~\vert~(\dd \Phi)(m)=0\}
\ee
the critical set of $\Phi$ and by:
\be
\cM_0\eqdef \cM\setminus \Crit\Phi
\ee
the {\em noncritical set} of $(\cM,\cG,\Phi)$. For any $E>0$,
we denote by:
\ben
\label{cME}
\cM(E)\eqdef \{m\in \cM~\vert~\Phi(m)<E\}
\een
the corresponding open subcritical level set of $\Phi$ and set:
\ben
\label{cM0E}
\cM_0(E)\eqdef \cM(E)\setminus \Crit \Phi~~.
\een

The multifield cosmological model parameterized by the quadruplet $\fM$ is obtained by assuming
that $g$ is an FLRW metric with flat spatial section:
\ben
\label{FLRW}
\dd s^2_g=-\dd t^2+a(t)^2\sum_{i=1}^3 \dd x_i^2
\een
(where $a$ is a smooth and strictly positive function) and that $\varphi$ depends only on the {\em
cosmological time} $t\eqdef x^0$. Define the {\em Hubble parameter}
through:
\be
H(t)\eqdef \frac{\dot{a}(t)}{a(t)}~~,
\ee
where the dot indicates derivation with respect to $t$.

\subsection{The cosmological equation}

\noindent When $H>0$ (which we assume throughout this paper), the
variational equations of \eqref{S} reduce to the {\em cosmological
equation}:
\ben
\label{eomsingle}
\nabla_t \dot{\varphi}(t)+\cH(\dot{\varphi}(t))\dot{\varphi}(t)+ (\grad_{\cG} \Phi)(\varphi(t))=0~~
\een
together with the condition:
\ben
\label{Hvarphi}
H(t)=H_\varphi(t)\eqdef \frac{1}{3}\cH(\dot{\varphi}(t))~~,
\een
where $\cH:T\cM\rightarrow \R_{>0}$ is the {\em rescaled Hubble function} of
$(\cM,\cG,\Phi)$, which is defined through:
\be
\cH(u)\eqdef \frac{1}{M_0}\left[||u||^2+2\Phi(\pi(u))\right]^{1/2}~~\forall u\in T\cM~~.
\ee
Here $\nabla_t\eqdef \nabla_{\dot{\varphi}(t)}$ and $\pi:T\cM\rightarrow \cM$ is
the bundle projection. The solutions
$\varphi:I\rightarrow \cM$ of \eqref{eomsingle} (where $I$ is a
non-degenerate interval) are called {\em cosmological curves}, while
their images in $\cM$ are called {\em cosmological orbits}.

A cosmological curve $\varphi:I\rightarrow \cM$ need not be an immersion.
Accordingly, we define the {\em singular and regular parameter sets}
of $\varphi$ through:
\beqan
&& I_\sing\eqdef \{t\in I~\vert~\dot{\varphi}(t)=0\}\subset I\\
&& I_\reg\eqdef I\setminus I_\sing= \{t\in I~\vert~\dot{\varphi}(t)\neq 0\}\subset I~~.
\eeqan
It can be shown that $I_\sing$ is an at most countable set.
The sets of {\em critical and noncritical times} of $\varphi$ are defined through:
\beqan
&& I_\crit\eqdef \{t\in I~\vert~\varphi(t)\in \Crit\Phi \}\nn\\
&& I_\noncrit \eqdef I\setminus I_\crit \eqdef \{t\in I~\vert~\varphi(t)\not\in \Crit\Phi\}~~.
\eeqan
The cosmological curve $\varphi$ is called {\em noncritical} if
$I_\crit=\emptyset$, i.e. if its orbit is contained in the non-critical set
$\cM_0$. It is easy to see that a cosmological curve $\varphi:I\rightarrow \cM$ 
is constant iff its orbit coincides with a critical point of $\Phi$, which
in turn happens iff there exists some $t\in I_\sing$ such that $\varphi(t)\in \Crit \Phi$.
Hence for any non-constant cosmological curve we have:
\be
I_\sing\cap I_\crit=\emptyset~~.
\ee

Given a cosmological curve $\varphi$, relation \eqref{Hvarphi}
determines $a$ up to a multiplicative constant. The cosmological
equation can be reduced to first order\cite{SLK} by passing to the tangent
bundle of $\cM$. This interpretation arises\cite{ren} by
viewing the second order time derivative $\ddot{\varphi}(t)$ appearing
in the expression:
\be
\nabla_t\dot{\varphi}^i(t)=\ddot{\varphi}(t)+\Gamma^i_{jk}(\varphi(t))\dot{\varphi}^j(t)\dot{\varphi}^k(t)
\ee
as an element of the double tangent bundle $TT\cM$ and the
opposite of the remaining terms of \eqref{eomsingle} as defining a
vector field $S\in \cX(T\cM)$ (called the {\em cosmological semispray}
or {\em second order vector field} of $(\cM,\cG,\Phi)$) which is
defined on $T\cM$. Then \eqref{eomsingle} is
equivalent with the integral curve equation of $S$. The flow of this
vector field on the total space of $T\cM$ is called the {\em
  cosmological flow} of the model.

\subsection{Cosmological observables}

\noindent For any $k\in \Z_{>0}$, let $\pi^k:J^k(\cM)\rightarrow \cM$ denote the bundle of $k$-order jets of curves in $\cM$.

\begin{definition}
An (off-shell) {\em scalar local cosmological
observable} of order $k\geq 0$ is a map $f:J^k(\cM)\rightarrow
\R$. The observable is called {\em basic} if it has order $k=1$; in
this case, $f$ is a real-valued function defined on $T\cM$.
\end{definition}

\noindent A simple example of basic cosmological observable is the
norm function $||~||:T\cM\rightarrow \R_{\geq 0}$ of $(\cM,\cG)$,
which is continuous everywhere but smooth only on the slit tangent
bundle ${\dot T}\cM$. Another example is the rescaled Hubble function
$\cH:T\cM\rightarrow \R_{\geq 0}$.

\begin{definition}
Given a cosmological observable $f$ of order $k$ and a 
curve $\varphi:I\rightarrow \cM$, the {\em evaluation} 
$f_\varphi:I\rightarrow \cM$ of $f$ along $\varphi$ is the function
defined though:
\be
f_\varphi \eqdef f\circ j^k(\varphi):I\rightarrow \R~~,
\ee
where $j^k(\varphi):I\rightarrow J^k(\cM)$ is the $k$-th jet
prolongation of $\varphi$. 
\end{definition}

\begin{remark}
Notice that a cosmological observable $f$ is completely determined by
its evaluation $f_\varphi$ on {\em arbitrary} curves $\varphi$ in
$\cM$. Accordingly, one can describe such an observable by giving its
evaluation on an arbitrary curve rather than using the jet bundle
description.
\end{remark}

The cosmological equation defines a codimension $d$ closed submanifold
$\fS$ (called the {\em cosmological shell}) of the total space of
$J^2(\cM)$, the second order jet bundle of curves of $\cM$ (which is a
bundle of rank $2d$ over $\cM$). Notice that $\fS$ is the image of a
section $\fs:J^1(\cM)\rightarrow J^2(\cM)$ of the natural surjection
$J^2(\cM)\rightarrow J^1(\cM)$ with local coordinate expression:
\be
\fs^i=-\Gamma^i_{jk}\dot{\varphi}^j\dot{\varphi}^k-\frac{1}{M_0} \left[\cG_{kl}(\varphi)\dot{\varphi}^k\dot{\varphi}^l+
  2\Phi(\varphi)\right]^{1/2}\dot{\varphi}^i- \cG^{ij}(\varphi)(\pd_j \Phi)(\varphi)~~.
\ee
In the jet bundle interpretation, the cosmological equation \eqref{eomsingle} reads:
\ben
\label{eoms}
\ddot{\varphi}^i=\fs^i(\varphi^1,\ldots, \varphi^d,\dot{\varphi}^1,\ldots,\dot{\varphi}^d)
\een
and amounts to the defining equations of $\fS$. We say that $\fs$ is
the {\em cosmological section} of $J^2(\cM)\rightarrow
J^1(\cM)$. Formal differentiation of \eqref{eoms} with respect to time defines
sections $\fs^{(k)}:J^k(\cM)\rightarrow J^{k+1}(\cM)$ of the projections
$J^{k+1}(\cM)\rightarrow J^k(\cM)$ for all $k\geq 1$, where
$\fs^{(1)}\eqdef \fs$. These determine the higher order prolongations
of \eqref{eoms}, whose iterative use allows us to express all formal
higher time derivatives of $\varphi$ in terms of $\varphi$ and
$\dot{\varphi}$. In turn, this gives sections
$\fs_k:J^1(\cM)\rightarrow J^{k+1}(\cM)$ of the natural maps
$J^{k+1}(\cM)\rightarrow J^1(\cM)$ for all $k\geq 1$.

\begin{definition}
The {\em dynamical reduction} (or {\em on-shell reduction}) of a
cosmological observable $f$ of order $k\geq 2$ is the basic
cosmological observable $f^\red:J^1(\cM)=T\cM\rightarrow \R$ defined
through:
\be
f^\red\eqdef f\circ \fs_{k-1}~~.
\ee
\end{definition}

\noindent The following statement is immediate:

\begin{prop}
For any cosmological curve $\varphi$, we have:
\be
f_\varphi=f^\red_\varphi=f^\red \circ \dot{\varphi}~~.
\ee
\end{prop}

One can also consider local observables of order $k$ defined on
non-empty open subsets $U\subset J^k(\cM)$; the considerations above
apply to this situation with very minor modifications. Finally,
one can define {\em vector local cosmological observables} of order
$k\in \Z_{>0}$ as local sections $\f$ of the {\em $k$-th Finsler
  bundle}:
\be
F^k\eqdef (\pi^k)^\ast(T\cM)\rightarrow J^k(\cM)
\ee
defined on a non-empty open set $U\subset J^k(\cM)$ (notice that
$F^1=F$ is the ordinary Finsler bundle of $\cM$). Such observables are
called {\em basic} if $k=1$, in which case $\f\in \Gamma(T\cM,F)$ can
be viewed as a map from an open subset $U$ of $T\cM$ to $T\cM$ which
satisfies $\f(u)\in T_{\pi(u)}\cM$ for all $u\in U$.  The evaluation
of a local vector observable $\f$ of order $k$ on a curve
$\varphi:I\rightarrow \cM$ whose $k$-th prolongation is contained in $U$ produces a section
$\f_\varphi\eqdef [j^k(\varphi)]^\ast(\f)\in
\Gamma(I,\varphi^\ast(T\cM))$, i.e. a map with associates to any $t\in
I$ a vector $\f_\varphi(t)\in T_{\varphi(t)}\cM$. The dynamical
reduction $\f^\red\eqdef (\fs_{k-1})^\ast(\f)$ is a local section of the Finsler bundle $F\rightarrow T\cM$ which
satisfies:
\be
\f_\varphi=(\dot{\varphi})^\ast(\f^\red)
\ee
when $\varphi$ is a cosmological curve. We will encounter such
vector observables later on.

\paragraph{The cosmological energy function and rescaled Hubble function.}

\

\begin{definition} The {\em cosmological energy function}
is the basic cosmological observable $E:T\cM\rightarrow \R_{>0}$ defined through:
\be
E(u)\eqdef \frac{1}{2}||u||^2+\Phi(\pi(u))~~\forall u\in T\cM~~.
\ee
The {\em cosmological
kinetic energy function} $E_\kin:T\cM\rightarrow \R_{\geq 0}$ and the
{\em cosmological potential energy function} $E_\pot:T\cM\rightarrow \R_{>0}$
are the basic cosmological observables defined through:
\ben
E_\kin(u)\eqdef \frac{1}{2}||u||^2~~,~~E_\pot(u)\eqdef \Phi(\pi(u))~~\forall u\in T\cM~~.
\een
\end{definition}

\noindent With these definitions, we have:
\be
E=E_\kin+E_\pot~~.
\ee
Notice that $E_\pot$ coincides with the natural lift of $\Phi$ to $T\cM$:
\be
E_\pot=\pi^\ast(\Phi)=\Phi\circ \pi~~.
\ee
Also notice the relation:
\be
\cH=\frac{1}{M_0}\sqrt{2E}~~.
\ee

\begin{prop}
The evaluation of the cosmological energy along any
cosmological curve $\varphi:I\rightarrow \cM$ satisfies the {\em
  cosmological dissipation equation}:
\ben
\label{DissipEq}
\frac{\dd E_\varphi(t)}{\dd t}=-\frac{\sqrt{2E_\varphi(t)}}{M_0}||\dot{\varphi}(t)||^2~~.
\een
\end{prop}

\begin{proof}
Follows immediately by using the cosmological equation \eqref{eomsingle}.
\end{proof}

Notice that the rescaled Hubble function $\cH:T\cM\rightarrow \R_{>0}$ is a basic cosmological
observable. The dissipation equation is equivalent with:
\ben
\label{dcH}
\frac{\dd \cH_\varphi(t)}{\dd t}=-||\dot{\varphi}(t)||^2
\een
for any cosmological curve $\varphi$. Integrating \eqref{dcH} from $t_1$ to $t_2$
along $\varphi$ gives:
\be
\cH_\varphi(t_1)-\cH_\varphi(t_2)=\int_{t_1}^{t_2}\dd t ||\dot{\varphi}(t)||^2
\ee
Since the integral in the right hand side is minimized by the shortest geodesic
which connects the points $\varphi(t_1)$ to $\varphi(t_2)$ in $\cM$, we have:
\be
\int_{t_1}^{t_2}\dd t ||\dot{\varphi}(t)||^2\geq \rd(\varphi(t_1),\varphi(t_2))^2~~,
\ee
where $\rd$ is the distance function of $(\cM,\cG)$. This gives the {\em Hubble inequality}:
\ben
\label{cHineq}
\cH_\varphi(t_1)-\cH_\varphi(t_2)\geq \rd(\varphi(t_1),\varphi(t_2))^2~~.
\een

\subsection{Some operations determined by the Finsler bundle of $\cM$}

\noindent Recall that $\pi:T\cM\rightarrow \cM$ denotes the bundle
projection. The following two fiber bundles defined on $T\cM$:
\be
F\eqdef \pi^\ast(T\cM)~~,~~\dot{F}\eqdef \pi^\ast(\dot{T}\cM)\subset F
\ee
are called respectively the {\em Finsler} and {\em slit Finsler
  bundle} of $T\cM$ (see \cite{SLK}). The $\pi$-pullback of $\cG$ (which we
denote by the same letter) makes $F$ into a Euclidean vector bundle, which contains
$\dot{F}$ as open fiber sub-bundle. The sections of $F$ defined over
a non-empty open subset $U\subset T\cM$ are called {\em Finsler vector fields}.
A Finsler vector field can be viewed as a map $f:U\rightarrow T\cM$
which satisfies the condition:
\be
\pi\circ f=\pi\vert_U~~,
\ee
i.e. which takes a vector tangent to $\cM$ at a point into vector
tangent to $\cM$ {\em at the same point}. The Finsler bundle
allows us to give a global geometric description to various 
operations on tangent vectors. 

\paragraph{The normalization map.}

The {\em  normalization map} of $(\cM,\cG)$ is the nowhere-vanishing Finsler
vector field $T\in \Gamma({\dot T}\cM, \dot{F})$ defined on
$\dot{T}\cM$ through:
\be
T(u)\eqdef \frac{u}{||u||} \in {\dot T} \cM~~\forall u\in \dot{T}\cM~~.
\ee

\begin{remark}
Given a curve $\varphi:I\rightarrow \cM$ and $t\in I_\reg$, the unit tangent vector
to $\varphi$ at time $t$ is given by:
\be
T_\varphi(t)=T(\dot{\varphi}(t))~~.
\ee
The map $T_\varphi:I_\reg\rightarrow T\cM$ identifies with the pull-back of $T$ through the
restriction of $\dot{\varphi}$ to $I_\reg$. 
\end{remark}

\paragraph{The parallel and normal projection of vector fields.}

For any vector field $X$ defined on $\cM$, let $X^\parallel, X^\perp\in
\Gamma({\dot T}\cM,F)$ be the Finsler vector fields defined on ${\dot
T}\cM$ through:
\beqa
&& X^\parallel(u)\eqdef \cG(T(u),X(\pi(u)))T(u)\in T_{\pi(u)}\cM\\
&& X^\perp(u)\eqdef X(\pi(u))-X^\parallel(u)\in T_{\pi(u)}\cM~~,
\eeqa
for all $u\in \dot{T}\cM$. For any $u\in \dot{T}\cM$, we have:
\be
\cG(X^\parallel(u),X^\perp(u))=0~~,~~X(u)=X^\parallel(u)+X^\perp(u)~~.
\ee
The vectors $X^\parallel(u)$ and $X^\perp(u)$ are respectively the
orthogonal projections of the vector $X(\pi(u))\in T_{\pi(u)}\cM$ on
$u$ and on the hyperplane $\Pi(u)\subset T_{\pi(u)}\cM$ which is
orthogonal to $u$ in the Euclidean vector space
$(T_{\pi(u)}\cM,\cG_{\pi(u)})$.

\section{Natural basic observables}
\label{sec:obs}

In this section, we give a geometric description of certain basic
cosmological observables which play a special role in the study of
multifield cosmological models. Some of these are defined directly on
$T\cM$, while others are obtained as the on-shell reduction of
second order observables. Since cosmological observables are
determined by their evaluation on arbitrary curves (see Section
\ref{sec:models}), we will use that description instead of the jet
bundle formulation.

\subsection{The first IR function of a scalar triple}

\noindent A first natural basic observable is the {\em first
  IR function} considered in \cite{ren}, which plays a crucial role in
the infrared dynamics of multifield cosmological models. This is
closely related to the {\em first slow roll
  function} (which is the on-shell reduction of a second order observable).

\begin{definition}
The {\em first IR function} of the scalar triple $(\cM,\cG,\Phi)$ is
the basic cosmological observable $\kappa:T\cM\rightarrow \R_{\geq 0}$
defined through:
\ben
\label{kappa}
\kappa(u)\eqdef \frac{E_\kin(u)}{E_\pot(u)}=\frac{||u||^2}{2\Phi(\pi(u))}~~\forall u\in T\cM~~.
\een
\end{definition}

\noindent Notice that $\kappa$ is continuous on
$T\cM$ but smooth only on the slit tangent bundle $\dot{T}\cM$. Also
notice that $\kappa(u)$ depends only on $\pi(u)\in \cM$ and $||u||\in
\R_{\geq 0}$, i.e. for all $m\in \cM$ we have:
\be
\kappa(u)=\kappa_m(||u||)~~\forall u\in T_m\cM~~,
\ee
where $\kappa_m:\R_{\geq 0}\rightarrow \R_{\geq 0}$ is the map defined
through:
\be
\kappa_m(x)\eqdef \frac{x^2}{2\Phi(m)}~~\forall x\in \R_{\geq 0}~~.
\ee

\paragraph{The first IR ``parameter'' of a cosmological curve.}

\

\begin{definition}
The {\em first IR ``parameter''} of the smooth curve $\varphi:I\rightarrow
\cM$ is the evaluation $\kappa_\varphi:I\rightarrow \R_{\geq 0}$ of $\kappa$ along $\varphi$:
\be
\kappa_{\varphi}(t)\eqdef \kappa(\dot{\varphi}(t))=\frac{||\dot{\varphi}(t)||^2}{2\Phi(\varphi(t))}~~.
\ee
\end{definition}

\noindent Notice that $\kappa_\varphi(t)$ vanishes when $t\in
I_\sing$, so $\kappa_\varphi$ is {\em not} a true parameter for $\varphi$
when $I_\sing\neq \emptyset$.  The first IR function allows us to
express the norm of any tangent vector $u\in T\cM$ as:
\ben
\label{kappanorm}
||u||=\sqrt{2E_\pot(u)\kappa(u)}=\sqrt{2\Phi(\pi(u))\kappa(u)}~~.
\een
For any $u\in \dot{T}\cM$, we have:
\be
T(u)=\frac{1}{\sqrt{2\Phi(\pi(u)) \kappa(u)}} \, u~~.
\ee
Also notice the relation:
\ben
\label{cHkappa}
\cH(u)=\frac{1}{M_0} \sqrt{2\Phi(\pi(u))}[1+\kappa(u)]^{1/2}~~.
\een

\paragraph{The first slow roll function of a scalar triple.}

\noindent It is traditional in cosmology to use a basic observable
which encodes the same information as $\kappa$ but is more convenient
for certain purposes. We give a geometric description of this below.

\begin{definition}
The {\em rescaled first slow roll function} of $(\cM,\cG,\Phi)$ is the
basic cosmological observable ${\hat \bepsilon}:T\cM\rightarrow
\R_{\geq 0}$ defined through:
\be
{\hat \bepsilon}(u)\eqdef \frac{E_\kin(u)}{E(u)}=\frac{\kappa(u)}{1+\kappa(u)}=\frac{||u||^2}{||u||^2+2\Phi(\pi(u))}~~\forall u\in T\cM~~,
\ee
while the {\em first slow roll function} of $(\cM,\cG,\Phi)$ is the
basic cosmological observable $\bepsilon:T\cM\rightarrow
\R_{\geq 0}$ given by:
\be
\bepsilon\eqdef 3{\hat \bepsilon}~~.
\ee
\end{definition}

\noindent We have:
\be
\hepsilon=f\circ \kappa~~,
\ee
where $f:\R_{\geq 0}\rightarrow [0,1)$ is the function defined through:
\be
f(x)\eqdef \frac{x}{1+x}~~.
\ee
This function has first derivative given by:
\be
f'(x)=\frac{1}{(1+x)^2}>0
\ee
and hence is strictly increasing. Notice the equivalence:
\be
f(x)=A\Longleftrightarrow x=\frac{A}{1-A}
\ee
and the relation:
\be
f(x)\approx x~~\mathrm{for}~x\ll 1~~.
\ee
We have ${\hat \bepsilon}=0$ iff $\kappa=0$ and $\bepsilon=1$
(i.e. ${\hat \bepsilon}=1/3$) iff $\kappa=1/2$. For $x>0$, the
inequality $\kappa(u)<x$ is equivalent with ${\hat
  \bepsilon}(u)<\frac{x}{1+x}$.

\paragraph{The first slow roll ``parameter'' of a cosmological curve.}

\noindent The traditional notion used in the cosmology literature is as follows:

\begin{definition}
The {\em first slow roll ``parameter''} of a curve
$\varphi:I\rightarrow \cM$ is the function
$\bepsilon_\varphi:I\rightarrow \R$ defined through:
\be
\bepsilon_{\varphi}(t)\eqdef -\frac{\dot{H_\varphi}(t)}{H_\varphi(t)^2}=-\frac{1}{H_\varphi(t)}\frac{\dd}{\dd t} \log H_\varphi(t)=3{\hat \bepsilon}_\varphi(t)~~,
\ee
where:
\be
{\hat \bepsilon}_\varphi(t)\eqdef -\frac{1}{\cH_\varphi(t)}\frac{\dd}{\dd t} \log \cH_\varphi(t)
\ee
is the {\em rescaled first slow roll ``parameter''} of $\varphi$.
\end{definition}

\begin{prop}
\label{prop:epsred}
Suppose that $\varphi:I\rightarrow \cM$ is a cosmological curve. Then
the (rescaled) first slow roll parameter of $\varphi$ coincides with
the on-shell reduction of the (rescaled) first slow roll function along
$\varphi$, i.e. we have:
\be
{\hat \bepsilon}_\varphi(t)={\hat \bepsilon}(\dot{\varphi}(t))~~\mathrm{and}~~ \bepsilon_\varphi(t)=\bepsilon(\dot{\varphi}(t))~~\forall t\in I~~.
\ee
\end{prop}

\begin{proof}
We have:
\be
\frac{\dd}{\dd t} \log \cH_\varphi(t)=\frac{\cG(\dot{\varphi}(t),\nabla_t \dot{\varphi}(t))+(\dd \Phi)(\varphi(t))(\dot{\varphi}(t))}{||\dot{\varphi}(t)||^2+2\Phi(\varphi(t))}=
-\frac{\cH_\varphi(t)||\dot{\varphi}(t)||^2}{||\dot{\varphi}(t)||^2+2\Phi(\varphi(t))}=-\frac{\cH_\varphi(t)\kappa_\varphi(t)}{1+\kappa_\varphi(t)}~~,
\ee
where in the second equality we used the cosmological equation and the
relation:
\be
\cG(\dot{\varphi}(t),(\grad \Phi)(\varphi(t)))=(\dd
\Phi)(\varphi(t))(\dot{\varphi}(t))~~.
\ee
This immediately implies the conclusion. 
\end{proof}

\begin{remark}
Since $\cH$ is a first order observable, the first slow roll parameter
of $\varphi$ is the evaluation along $\varphi$ of a {\em second order}
observable (the ``first slow roll observable''), whose on-shell
reduction coincides with the first slow roll function by Proposition
\ref{prop:epsred}. Hence the first IR parameter encodes the same information
as the on-shell reduction of the second order first slow roll observable.
\end{remark}

\noindent Notice that $\hepsilon_\varphi(t)$ vanishes when $t\in
I_\sing$, so it is {\em not} a true parameter for
$\varphi$ when $I_\sing \neq \emptyset$. The condition for inflation
can be formalized as follows (see \cite{modular}). 

\paragraph{The inflation region.}

\

\begin{definition}
The {\em inflation region} of $(\cM,\cG,\Phi)$ is the subset of $T\cM$ defined through:
\be
\cR(\cM,\cG,\Phi)\eqdef \{u\in T\cM~\vert~||u||^2<\Phi(\pi(u))\}=\{u\in T\cM~\vert~\kappa(u)<1/2 \}=\{u\in T\cM~\vert~\bepsilon(u)<1 \}~~.
\ee
A cosmological curve $\varphi:I\rightarrow \cM$ is called {\em
inflationary at time $t\in I$} if $\dot{\varphi}(t)\in
\cR(\cM,\cG,\Phi)$, which amounts to the usual
condition for inflation:
\be
\bepsilon_\varphi(t)<1\Longleftrightarrow {\hat \bepsilon}_\varphi(t)<1/3\Longleftrightarrow \kappa_\varphi(t)<1/2~~.
\ee
The cosmological curve is called {\em inflationary on the non-empty
interval $J\subset I$} if $\varphi$ is inflationary for all $t\in J$.
\end{definition}

\noindent Notice that $\cR(\cM,\cG,\Phi)$ is an open tubular neighborhood of
the zero section of $T\cM$. For any cosmological curve
$\varphi:I\rightarrow \cM$, the set of its inflationary times:
\be
I_\infl\eqdef \varphi^{-1}(\cR(\cM,\cG,\Phi))\subset I
\ee
is a (possibly empty) relatively open subset of $I$ and hence it is the
intersection of $I$ with an at most countable disjoint union of open
intervals of the real axis.

\paragraph{The first slow roll condition with parameter $\epsilon$.}

\

\begin{definition}
Let $\epsilon\in (0,1]$. We say that a vector $u\in T\cM$ satisfies the {\em first slow roll
condition} with parameter $\epsilon$ if:
\be
\kappa(u)<\epsilon~~,~~\mathrm{i.e.}~~||u||^2< 2\Phi(\pi(u))\, \epsilon~~.
\ee
The {\em first slow roll region} of $(\cM,\cG,\Phi)$ at parameter
$\epsilon$ is the open subset of $T\cM$ defined through:
\be
\cS^1_\epsilon(\cM,\cG,\Phi)\eqdef \{u\in T\cM~\vert~\kappa(u)<\epsilon\}=\{u\in T\cM~\vert~||u||< \sqrt{2\epsilon \Phi(\pi(u))}\}\subset T\cM~~.
\ee
\end{definition}

\noindent Since $\Phi>0$, the image of the zero section of $T\cM$ is
contained in $\cS^1_\epsilon(\cM,\cG,\Phi)$ and hence we have:
\ben
\label{picS1}
\pi(\cS^1_\epsilon(\cM,\cG,\Phi))=\cM~~.
\een

\begin{definition}
We say that a cosmological curve $\varphi:I\rightarrow \cM$ satisfies
the first slow roll condition with parameter $\epsilon\in (0,1]$ at
time $t\in I$ if $\dot{\varphi}(t)\in \cS^1_\epsilon(\cM,\cG,\Phi)$.
\end{definition}

Since $\cS^1_\epsilon(\cM,\cG,\Phi)$ is an open subset of $T\cM$ and
the canonical lift $\dot{\varphi}:I\rightarrow T\cM$ of $\varphi$ is
continuous, the set:
\be
I_\epsilon^1\eqdef \dot{\varphi}^{-1}(\cS^1_\epsilon(\cM,\cG,\Phi))\subset I
\ee
is an open subset of $I$. When non-empty, the set $I^1_\epsilon$ is an at most
countable union of relatively open subintervals of $I$.

\subsection{The conservative function of a scalar triple}

\noindent A second natural basic observable is the {\em conservative
  function}, which controls the {\em conservative
  approximation} discussed in Section \ref{sec:cons}. This is
defined as the inverse of the norm of a Finsler vector
field which is naturally associated to every scalar triple.

\begin{definition}
\label{def:relgrad}
The {\em relative gradient field} of $(\cM,\cG,\Phi)$ is the Finsler
vector field $\q\in \Gamma({\dot T}\cM,F)$ defined on $\dot{T}\cM$
through:
\ben
\label{q}
\q(u)\eqdef \frac{(\grad \Phi)(\pi(u))}{\cH(u) ||u||}= M_0 \frac{(\grad \Phi)(\pi(u))}{||u|| [||u||^2+2\Phi(\pi(u))]^{1/2}}~~.
\een
\end{definition}

\noindent Notice that $\q(u)=0$ iff $\pi(u)\in \Crit\Phi$. The
relative gradient field is a basic {\em vector} local observable in
the sense of Section \ref{sec:models}.

\begin{definition}
The {\em characteristic one-form} $\Xi\in \Omega^1(\cM)$ of the model is defined through:
\ben
\label{omega}
\Xi\eqdef\frac{M_0}{2\Phi} \dd\Phi=\frac{M_0}{2}\dd \log\Phi~~.
\een
\end{definition}

\noindent The pointwise norm $||\Xi||$ vanishes exactly on the critical set $\Crit\Phi$.
For any $u\in \dot{T}\cM_0$, we have:
\be
\q(u)= \frac{||\Xi(\pi(u))||}{[\kappa(u)(1+\kappa(u)]^{1/2}}n_\Phi(\pi(u))~~,
\ee
where $n_\Phi\in \cX(\cM_0)$ is the {\em normalized gradient field} of $\Phi$:
\be
n_\Phi(m)\eqdef \frac{(\grad\Phi)(m)}{||(\dd\Phi)(m)||}~~\forall m\in\cM_0~~.
\ee

\begin{definition}
\label{def:cons}
The {\em conservative function} of $(\cM,\cG,\Phi)$ is the function
$c:T\cM_0\rightarrow \R_{\geq 0}$ defined through:
\ben
\label{cdef1}
c(u)\eqdef \frac{1}{||\q(u)||}=\frac{1}{M_0}\frac{||u||\left[||u||^2+2\Phi(\pi(u))\right]^{1/2}}{||(\dd \Phi)(\pi(u))||}~~\forall u\in T\cM_0~~,
\een
i.e.: 
\ben
\label{ckappaomega}
c(u)=\frac{ \left[\kappa(u)(1+\kappa(u))\right]^{1/2}}{||\Xi(\pi(u))||}~~\forall u\in T\cM_0~~.
\een
\end{definition}

\noindent Thus:
\be
\q(u)=\frac{1}{c(u)}n_\Phi(\pi(u))~~\forall u\in T\cM_0~~.
\ee
Notice that $c(u)$ depends only on $m\eqdef \pi(u)\in \cM_0$
and on $||u||$, i.e. we have:
\be
c(u)=c_m(||u||)~~\forall u\in T_m\cM_0~~,
\ee
where $c_m:\R_{\geq 0}\rightarrow \R_{\geq 0}$ is the map defined through:
\be
c_m(x)\eqdef  \frac{1}{M_0}\frac{x\left[x^2+2\Phi(m)\right]^{1/2}}{||(\dd \Phi)(m)||}~~.
\ee

\begin{remark}
\label{rem:f}
Let $\Lambda>0$. The function $f:\R\rightarrow \R$ defined through
$f(x)=x(1+x)$ has a minimum at $x=-1/2$ and is strictly increasing
when $x>-1/2$. Thus $f(\R_+)=[0,+\infty)$ and a non-negative $x$
satisfies $f(x)\leq \Lambda$ iff $x\leq x_+$ where: \be
x_+=\frac{1}{2}[-1+\sqrt{1+4\Lambda}] \ee is the non-negative solution
of the equation $f(x)=\Lambda$.
\end{remark}

The first IR function $\kappa$ and the conservative function $c$
determine each other through relation \eqref{ckappaomega}. Using
Remark \ref{rem:f}, this relation can be solved for $\kappa(u)$ as:
\ben
\label{kappau}
\kappa(u)=\frac{1}{2}\left[-1+\sqrt{1+4c(u)^2 ||\Xi(\pi(u))||^2}\right]~~.
\een
In particular, we have $\kappa(u)\ll 1$ iff $c(u)||\Xi(\pi(u))||\ll
1$. When this condition is satisfied, relation \eqref{kappau} becomes
$\kappa(u)\approx c(u)^2 ||\Xi(\pi(u))||^2$.

\paragraph{The conservative ``parameter'' of a non-constant cosmological curve.}

\

\begin{definition}
The {\em conservative ``parameter''} of a non-constant cosmological curve
$\varphi:I\rightarrow \cM_0$ is the function $c_\varphi:I\rightarrow
\R_{\geq 0}\cup\{+\infty\}$ defined through:
\be
c_\varphi(t)\eqdef \frac{\cH(\dot{\varphi}(t)) ||\dot{\varphi}(t)||}{||(\dd \Phi)(\varphi(t))||}=\frac{1}{M_0} \frac{||\dot{\varphi}(t)||[||\dot{\varphi}(t)||^2+2\Phi(\varphi(t))]^{1/2}}{||(\dd \Phi)(\varphi(t))||}=
c(\dot{\varphi}(t))~~\forall t\in I~~,
\ee
which gives the evaluation of $c$ along $\varphi$.
\end{definition}

\noindent Notice that $c_\varphi$ is well-defined since the speed of a
non-constant cosmological curve cannot vanish at a critical time. We
have:
\be
c_\varphi(t)=+\infty \Longleftrightarrow t\in I_\crit~~.
\ee

\noindent For non-constant curves $\varphi$, the cosmological equation \eqref{eomsingle} reads:
\be
\nabla_t \dot{\varphi}(t)+ (\grad_{\cG} \Phi)(\varphi(t))=-||(\dd \Phi)(\varphi(t))|| c_\varphi(t) T(\dot{\varphi}(t))~~.
\ee
The {\em conservative approximation} consists of neglecting the right
hand side, which replaces \eqref{eomsingle} by the {\em conservative
equation}:
\be
\nabla_t \dot{\varphi}(t)+ (\grad_{\cG} \Phi)(\varphi(t))=0~~.
\ee
This approximation is discussed in Subsection \ref{sec:cons}.
The conservative approximation is accurate in the
{\em quasi-conservative regime} $c_\varphi(t)\ll 1$. When $c_\varphi(t)$ is large,
one can neglect the friction term instead, which amounts to
replacing \eqref{eomsingle} with the {\em dissipative equation}:
\be
\nabla_t \dot{\varphi}(t)+ \cH_\varphi(t)\dot{\varphi}(t)=0~~.
\ee
This defines the {\em dissipative approximation} (studied in Subsection
\ref{sec:friction}), which is accurate in the {\em strongly dissipative
regime} $c_\varphi(t)\gg 1$.

\begin{definition}
For any $\epsilon\in (0,1]$, the {\em conservative region }
$C_\epsilon(\cM,\cG,\Phi)$ of $(\cM,\cG,\Phi)$ is the open tubular neighborhood of
the zero section of $T\cM_0$ defined through:
\be
C_\epsilon(\cM,\cG,\Phi)\eqdef \{u\in T\cM_0\vert c(u)<\epsilon\}\subset T\cM_0~~.
\ee
We say that a vector $u\in T\cM_0$ satisfies the {\em conservative
condition} with parameter $\epsilon$ if $u\in C_\epsilon$, i.e. if:
\ben
\label{ucons1}
c(u)<\epsilon~~.
\een
We say that a cosmological curve $\varphi$ satisfies the conservative condition with
parameter $\epsilon\in (0,1]$ at time $t\in I$ if $\dot{\varphi}(t)\in
C_\epsilon(\cM,\cG,\Phi)$.
\end{definition}

\noindent  By Remark \ref{rem:f}, the conservative condition
\eqref{ucons1} is equivalent with:
\ben
\label{ucons2}
\kappa(u)<\frac{1}{2}\left[-1+\sqrt{1+\epsilon^2 M_0^2 \frac{||(\dd \Phi)(\pi(u))||^2}{\Phi(\pi(u))^2}}~\right]~~,
\een
i.e. with:
\ben
\label{ucons3}
||u||<A_\epsilon(\pi(u))
\een
where:
\ben
\label{Adef}
A_\epsilon\eqdef \sqrt{-\Phi+\sqrt{\Phi^2+\epsilon^2 M_0^2 ||\dd \Phi||^2}}~~
\een
is the {\em conservative bound function} of $(\cM,\cG,\Phi)$ at parameter $\epsilon$.

\begin{definition}
The {\em conservative closure} of $(\cM,\cG,\Phi)$ is the region:
\be
\overline{C_1}=\{u\in T\cM_0~\vert~c(u)\leq 1\}~~.
\ee
\end{definition}

\subsection{The tangential acceleration, characteristic angle and turning rate}

\noindent We next discuss further natural scalar and vector basic observables which are
on shell reductions of corresponding second order observables. For this, we start from a
natural vector observable of second order which we describe through its evaluation on arbitrary
curves (rather than in jet bundle language). Some of the following definitions are standard except for
certain conventional scale factors which we find convenient to eliminate in order to simplify
various formulas.

\begin{definition}
The {\em opposite relative acceleration} of a smooth curve
$\varphi:I\rightarrow \cM$ is the smooth function
$\veta_\varphi:I_\reg\rightarrow T\cM$ defined through:
\be
\veta_\varphi(t)=-\frac{1}{H_\varphi(t)}\frac{\nabla_t \dot{\varphi}(t)}{||\dot{\varphi}(t)||}=3{\hat \veta}_\varphi(t)~~\forall t\in I_\reg~~,
\ee
where:
\ben
\label{hatvetadef}
{\hat \veta}_\varphi(t)\eqdef -\frac{1}{\cH_\varphi(t)}\frac{\nabla_t \dot{\varphi}(t)}{||\dot{\varphi}(t)||}~~(t\in I_\reg)
\een
is the {\em rescaled acceleration} of $\varphi$.
\end{definition}

\begin{remark}
In reference \cite{genalpha}, the opposite relative acceleration
$\veta_\varphi(t)$ was called the {\em vector gradient flow parameter}
of $\varphi$ at time $t$, while its norm $||\veta_{\varphi}(t)||$ was called
the {\em scalar gradient flow parameter}.
\end{remark}

\begin{definition}
The {\em rescaled tangential acceleration} (a.k.a. {\em rescaled second
slow roll parameter}) of a smooth curve $\varphi:I\rightarrow \cM$ is
the function ${\hat \eta}^\parallel_\varphi:I_\reg\rightarrow \R$
obtained by projecting ${\hat \veta}_\varphi(t)$ on the unit tangent
vector $T_\varphi(t)$ to $\varphi$ at time $t\in I_\reg$:
\be
\hetap_\varphi(t)\eqdef\cG({\hat \veta}_\varphi(t),T_\varphi(t))~~\forall t\in I_\reg~~.
\ee
The {\em rescaled normal acceleration} of $\varphi$ is the function
${\hat \veta}^\perp_\varphi:I_\reg\rightarrow T\cM$ obtained by projecting
${\hat \veta}_\varphi(t)$ on the hyperplane $\Pi_\varphi(t)\eqdef
T_\varphi(t)^\perp\subset T_{\varphi(t)}\cM$ normal to
$T_\varphi(t)$ inside the tangent space $T_{\varphi(t)}\cM$:
\ben
\label{hatvetaperp}
{\hat \veta}_\varphi^\perp(t)\eqdef {\hat \veta}_\varphi(t)-\hetap_\varphi(t)T_\varphi(t)~~\forall t\in I_\reg~~.
\een
The {\em characteristic angle} $\theta_\varphi(t)\in [0,\pi]$ of
$\varphi$ at time $t\in I_\reg\cap I_\noncrit$ is the angle between
$\dot{\varphi}(t)$ and $(\grad\Phi)(\varphi(t))$. For any $t\in  I_\reg\cap I_\noncrit$,
we have:
\be
\hetap_\varphi(t)=||\hveta_\varphi(t)|| \cos\theta_\varphi(t)~~.
\ee
\end{definition}

\begin{remark}
Let $s$ be an increasing arc length parameter on the curve $\varphi$, thus $\dot{s}=||\dot{\varphi}||$.
We have:
\ben
\label{compacc}
[\nabla_t\dot{\varphi}(t)]^\parallel=\frac{\dd}{\dd t}||\dot{\varphi}(t)||=\ddot{s}~~,~~[\nabla_t\dot{\varphi}(t)]^\perp=\nabla_t\dot{\varphi}(t)-\left(\frac{\dd}{\dd t}\log ||\dot{\varphi}(t)||\right)\dot{\varphi}(t)~~,
\een
where we used the relation:
\be
\cG(\nabla_t \dot{\varphi}(t), \dot{\varphi}(t))=\frac{1}{2}\frac{\dd}{\dd t} ||\dot{\varphi}(t)||^2=||\dot{\varphi}(t)|| \, \frac{\dd}{\dd t} ||\dot{\varphi}(t)||~~,
\ee
which follows from the fact that $\cG$ is covariantly-constant. In particular, we have:
\be
\hetap_\varphi(t)=-\frac{1}{\cH_\varphi(t)}\frac{\dd}{\dd t}\log ||\dot{\varphi}(t)||=-
M_0\frac{\frac{\dd}{\dd t}\log ||\dot{\varphi}(t)||}{\sqrt{2\Phi(\varphi(t))}[1+\kappa_{\varphi}(t)]^{1/2}}~~\forall t\in I_\reg~~.
\ee
\end{remark}

\begin{definition}
The vector:
\ben
\label{bOmega}
\bOmega_\varphi(t)=-\nabla_tT_\varphi(t)\in T_{\varphi(t)}\cM~~(t\in I_\reg).
\een
is called the {\em turning rate vector} of $\varphi$ at time $t$.
\end{definition}

Notice that $\bOmega_\varphi(t)\perp T_{\varphi}(t)$.
Since $\dot{\varphi}(t)^\perp=0$, the defining relation \eqref{bOmega}
can be written as:
\be
\bOmega_\varphi(t)=-\frac{[\nabla_t\dot{\varphi}(t)]^\perp}{||\dot{\varphi}(t)||}~~.
\ee
The Frenet-Serret formulas give:
\be
\nabla_t T_{\varphi}(t)=\chi_\varphi(t) ||\dot{\varphi}(t)|| n_\varphi(t)~~,
\ee
where:
\be
n_\varphi(t)\eqdef \frac{N_\varphi(t)}{||N_\varphi(t)||}
\ee
with:
\be
N_\varphi(t)\eqdef (\nabla_t\dot{\varphi}_t)^\perp
\ee
and $\chi_\varphi(t)\eqdef ||\nabla_t\dot{\varphi}_t|| \in \R_{\geq
0}$ is the principal curvature of $\varphi$ at time $t$. Thus:
\be
\bOmega_\varphi(t)=\Omega_\varphi(t) n_\varphi(t)
\ee
where:
\be
\Omega_\varphi(t)\eqdef -\chi_\varphi(t)||\dot{\varphi}(t)||~~
\ee
is the {\em  scalar turning rate} of $\varphi$ at time $t$. In particular, we have:
\be
||\bOmega_\varphi(t)||=|\Omega_\varphi(t)|=|\chi_\varphi(t)| \, \, ||\dot{\varphi}(t)||~~.
\ee
In terms of $\bOmega_\varphi$, relation \eqref{hatvetaperp} takes the form:
\be
{\hat \veta}_\varphi^\perp(t)=\frac{\bOmega_\varphi(t)}{\cH_\varphi(t)}~~.
\ee
We have:
\be
\veta_\varphi(t)=\hetap_\varphi(t)T_\varphi(t)+\heta^\perp_\varphi(t)~~\forall t\in I_\reg
\ee
and hence:
\ben
\label{nablaheta}
\nabla_t\dot{\varphi}(t)=-\cH_\varphi(t) \left[\hetap_\varphi(t)\dot{\varphi}(t)+||\dot{\varphi}(t)||\heta^\perp_\varphi(t)\right]~~\forall t\in I_\reg~~.
\een

\paragraph{The rescaled acceleration field of $(\cM,\cG,\Phi)$.}

\noindent We next define a Finsler vector field on $\dot{T}\cM$ which
encodes the rescaled acceleration of all {\em cosmological} curves. It
is the on-shell reduction of the second order rescaled acceleration observable.

\begin{definition}
The {\em rescaled acceleration field} of $(\cM,\cG,\Phi)$ is the
Finsler vector field ${\hat \veta}\in \Gamma({\dot T}\cM,F\cM)$
defined on ${\dot T}\cM$ through:
\ben
\label{veta}
{\hat \veta}(u)\eqdef n(u)+\q(u)=n(u)+ \frac{\Xi(\pi(u))(n(u))}{[\kappa(u)(1+\kappa(u))]^{1/2}} n_\Phi(\pi(u))=n(u)+\frac{1}{c(u)}n_\Phi(\pi(u))~~,
\een
where $\q$ is the relative gradient field of $(\cM,\cG,\Phi)$ (see Definition \ref{def:relgrad}).
\end{definition}

\noindent Notice that the rescaled acceleration field determines the conservative function:
\be
c(u)=\frac{1}{||\q(u)||}=\frac{1}{||{\hat \veta}(u)-n(u)||}~~\forall u\in \dot{T}\cM_0~~.
\ee

\paragraph{The rescaled tangential acceleration function and normal acceleration field of $(\cM,\cG,\Phi)$.}

\

\begin{definition}
\label{def:etafunctions}
The {\em rescaled tangential acceleration function} (a.k.a. ``second slow roll function'') $\hetap:{\dot
T}\cM\rightarrow \R$ is the basic cosmological observable defined
through:
\ben
\label{etaomega}
\hetap(u)=\cG({\hat \veta}(u),n(u))=1+\frac{(\dd\Phi)(\pi(u))(u)}{\cH(u)||u||^2}=1+ \frac{\cos\theta(u)}{c(u)}~~.
\een
The {\em rescaled normal acceleration field} ${\hat \veta}^\perp\in
\Gamma({\dot T}\cM, F\cM)$ of $(\cM,\cG,\Phi)$ is the Finsler vector
field defined on ${\dot T}\cM$ through:
\ben
\label{hvetaperp}
{\hat \veta}^\perp(u)=\q^\perp (u)=\frac{(\grad\Phi)^\perp(u)}{\cH(u)||u||}=\frac{n_\Phi^\perp(u)}{c(u)}~~.
\een
\end{definition}

\noindent We have:
\be
\veta(u)=\hetap(u)n(u)+\heta^\perp(u)~~\forall u\in \dot{T}\cM~~.
\ee
Notice that $\hetap(u)=1$ and ${\hat
  \veta}^\perp(u)=0$ when $\pi(u)\in \Crit\Phi$.

\paragraph{The turning rate field of $(\cM,\cG,\Phi)$.}

\

\begin{definition}
The {\em turning rate field} of $(\cM,\cG,\Phi)$ is the Finsler vector field
$\bOmega\in \Gamma({\dot T}\cM,F\cM)$ given by:
\be
\bOmega(u)\eqdef \frac{(\grad \Phi)^\perp(u)}{||u||}~~(u\in \dot{T}\cM)~~.
\ee
\end{definition}

\noindent Relation \eqref{hvetaperp} reads:
\ben
\label{hetapOmega}
{\hat \veta}^\perp(u)=\q^\perp(u)=\frac{\bOmega(u)}{\cH(u)}~~.
\een

The following result shows that the basic scalar and vector observables introduced above are the
on-shell reductions of the corresponding second order observables. 

\begin{prop}
\label{prop:veta}
Suppose that $\varphi:I\rightarrow \cM$ is a cosmological curve. Then
for any $t\in I_\reg$ we have:
\be
{\hat \veta}_\varphi(t)={\hat \veta}(\dot{\varphi}(t))~~,~~\hetap_\varphi(t)=\hetap(\dot{\varphi}(t))~~,~~{\hat \veta}^\perp_\varphi(t)={\hat \veta}^\perp(\dot{\varphi}(t))~~,~~\bOmega_\varphi(t)=\bOmega(\dot{\varphi(t)})~~.
\ee
\end{prop}

\begin{proof}
The rescaled acceleration of $\varphi$ can be expressed as follows
using the cosmological equation:
\be
{\hat \veta}_\varphi(t)=\frac{\cH_\varphi(t)\dot{\varphi}(t)+(\grad\Phi)(\varphi(t))}{\cH_\varphi(t)||\dot{\varphi}(t)||}={\hat \veta}(\dot{\varphi}(t))~~.
\ee
Thus:
\ben
\label{hetarel1}
\hetap_\varphi(t)=1+\frac{(\dd\Phi)(\varphi(t))(\dot{\varphi}(t))}{\cH_\varphi(t)||\dot{\varphi}(t)||^2}=1+M_0\frac{(\dd\Phi)(\varphi(t))(\dot{\varphi}(t))}{(2\Phi(\varphi(t)))^{3/2}
  \kappa_{\varphi}(t)[1+\kappa_{\varphi}(t)]^{1/2}}
\een
and:
\ben
\label{hetarel2}
{\hat \veta}^\perp_\varphi(t)=\frac{(\grad\Phi)^\perp(\varphi(t))}{\cH_\varphi(t)||\dot{\varphi}(t)||}=M_0 \frac{(\grad\Phi)^\perp(\varphi(t))}{2\Phi(\varphi(t))[\kappa_{\varphi}(t)(1+\kappa_{\varphi}(t))]^{1/2}}~~,
\een
where we used \eqref{kappanorm}. This immediately gives the conclusion.
\end{proof}

Summarizing, we consider four natural basic observables, namely
$\kappa$, $c$, $\hetap$ and $\bOmega$, the last of which is a vector
observable. We also introduce an auxiliary basic observable
$\theta$:

\begin{definition}
The {\em characteristic angle function} of $(\cM,\cG,\Phi)$ is the map $\theta:\dot{T}\cM_0\rightarrow [0,\pi)$ defined through:
\be
\cos\theta(u)=\cG(T(u),n_\Phi(\pi(u)))~~\forall u\in \dot{T}\cM~~,
\ee
where $n_\Phi\in \cX(\cM_0)$ is the normalized gradient field of $(\cM,\cG,\Phi)$. 
\end{definition}

\noindent The following statement is immediate:

\begin{prop}
For any curve $\varphi:I\rightarrow \cM$ and any $t\in I_\reg\cap I_\noncrit$, we have:
\be
\theta_\varphi(t)=\theta(\dot{\varphi}(t))~~,
\ee
where $\theta_\varphi$ is the characteristic angle of $\varphi$. 
\end{prop}

\subsection{Reconstruction of the scalar field metric and potential from $\kappa$ and $\hetap$}

\noindent Knowledge of the first and second slow-roll
functions (equivalently, knowledge of the first IR function and of the
second slow-roll function) allows one to reconstruct the positive
homothety class of the pair $(\cG,\Phi)$ when $M_0$ is fixed.
Indeed, relation \eqref{etaomega} gives:
\ben
\label{omegaeta}
\Xi(\pi(u))(n(u))= [\kappa(u)(1+\kappa(u))]^{1/2} \left[-1+\hetap(u)\right]~~\forall u\in {\dot T}\cM~~,
\een
showing that $\kappa$ and $\hetap$ determine
$\Xi$. Since $\cM$ is connected, this implies that $\kappa$ and
$\hetap$ determine the scalar potential $\Phi$ up to a
multiplicative constant $A>0$. Now relation \eqref{kappanorm}
determines $\cG$ up to multiplication by $A$.  Thus $\kappa$ and
$\eta^{\parallel}$ determine the positive homothety class of the pair
$(\cG,\Phi)$. In particular, the rescaled first and second slow-roll
functions encode the same information as this homothety class.

\subsection{Geometric interpretation of $\hetap$}
\label{subsec:etageom}

\noindent Relation \eqref{etaomega} can
be written as:
\ben
\label{costheta}
\cos\theta(u)=c(u)[\hetap(u)-1]
\een
and requires:
\be
\hetap(u)\in \left[1-\frac{1}{c(u)},1+\frac{1}{c(u)}\right]~~.
\ee
In particular, the interval within which $\hetap(u)$
can take values is centered on $1$ and constrained by the value of
$c(u)$ and hence by the norm of $u$. This interval is very large in
the quasi-conservative regime $c(u)\ll 1$ and becomes narrow in the
strongly dissipative regime $c(u)\gg 1$, when $\hetap(u)$
is forced to be close to one (see Figure \ref{fig:ceta}). In particular,
the generalized ``ultra slow roll'' approximation $\hetap\approx 1$
is accurate in the strongly dissipative regime.

\begin{figure}[H] \centering
\includegraphics[width=.7\linewidth]{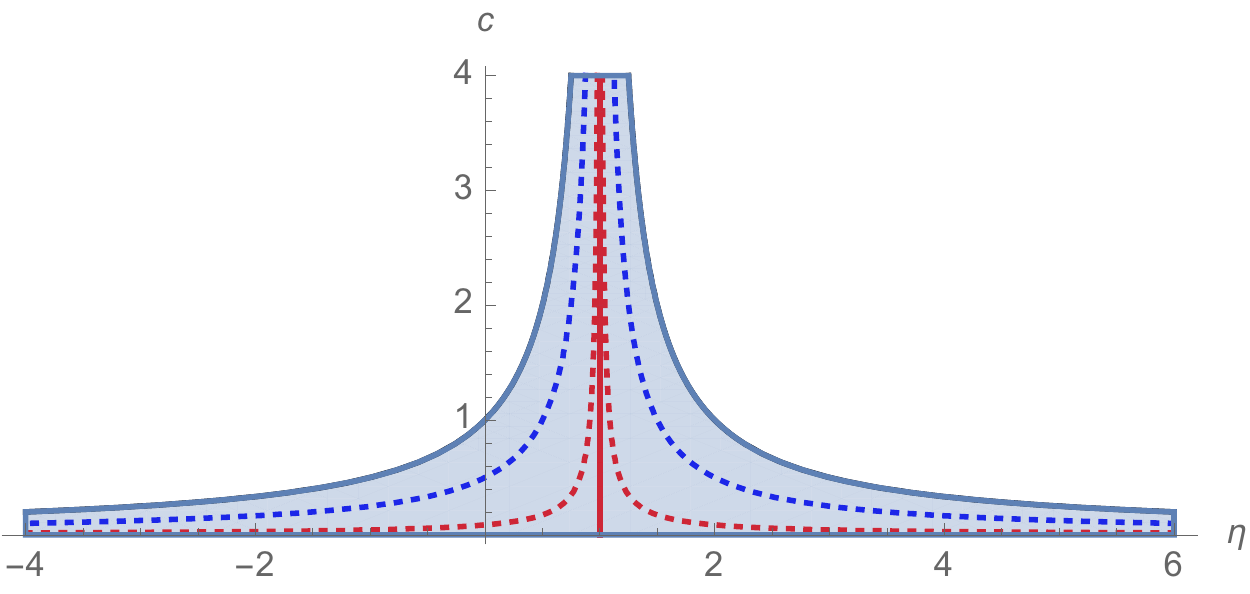}
\caption{Admissible domain of $c$ and $\hetap$. The
right and left boundaries of the domain are the hyperbolas $c({\hat
\eta}^\parallel-1)=-1$ and $c(\hetap-1)=+1$, which
correspond respectively to $\theta=0$ and $\theta=\pi$. The vertical
red line in the middle has equation $\hetap=1$ and
corresponds to $\theta=\pi/2$. The dotted blue hyperbolas correspond
to $\theta=\pi/3$ and $\theta=2\pi/3$ while the dotted red hyperbolas
correspond to $\theta=4\pi/9$ and $\theta=5\pi/9$.  The interval
within which $\hetap$ can vary for a fixed value of $c$
is obtained by intersecting the corresponding horizontal line with the
domain shown in the figure. This interval is centered on the value
$\hetap=1$ and its length decreases as $c$
increases. Accordingly, the strongly dissipative regime $c\gg 1$ forces ${\hat
\eta}^\parallel$ to be close to one and hence the generalized ``ultra slow roll''
approximation is accurate in this regime.}
\label{fig:ceta}
\end{figure}
    
Notice that $|\eta^\parallel(u)|\ll 1$ iff $\cos\theta(u)\approx -c(u)$,
which in particular requires that $u$ points towards decreasing values of $\Phi$.
On the other hand, \eqref{kappau} gives:
\ben
\label{unorm}
||u||=\sqrt{\Phi(\pi(u))\left[-1+\sqrt{1+4c(u)^2 ||\Xi(\pi(u))||^2}\right]}~~.
\een
Let $\pi(u)=m\in \cM_0$. When $\theta(u)=\pi/2$, relation
\eqref{costheta} requires $\hetap(u)=1$.  In this case,
$u$ is orthogonal to $n_\Phi(u)$ in $T_m\cM$ and its norm determines
and is determined by $c(u)$ through relation \eqref{unorm}. This is
the so-called ``ultra slow roll'' case, which generalizes the ultra slow roll regime
of one-field models. By
the remarks above, the strongly dissipative regime $c(u)\gg 1$ forces
$\hetap(u)\approx 1$ and hence a generalized ``ultra slow roll''
approximation is accurate in this regime.

Suppose now that $\theta(u)\neq \pi/2$, i.e. ${\hat
\eta}^\parallel(u)\neq 1$. Then eliminating $c(u)$ from
\eqref{costheta} and substituting in \eqref{unorm} gives:
\be
||u||=\sqrt{\Phi(\pi(u))\left[-1+\sqrt{1+\frac{4 ||\Xi(\pi(u))||^2}{(1-\hetap(u))^2} \cos^2\theta(u)}\right]}~~.
\ee
When $\hetap(u)$ is fixed, this determines $||u||$ as a
function of $\theta(u)$, where the sign of $\cos\theta(u)$ must equal
that of $\hetap(u)-1$ by \eqref{costheta} since $c(u)$ is
positive. Thus fixing the value of $\hetap$ forces $u$ to lie within a
hypersurface of revolution with axis given by the line determined by $(\grad\Phi)(m)$ inside $T_m\cM$. This hypersurface passes
through the origin of $T_m\cM$ at $\theta=0$ and cuts the line
determined by $(\grad\Phi)(m)$ again at a point corresponding to the
maximal norm of $u$, which is attained for $\theta=0$ or $\theta=\pi$,
depending on the sign of $\hetap-1$. This maximal value
of $||u||$ is given by:
\be
||u||_{\mathrm{max}}=\sqrt{\Phi(m)\left[-1+\sqrt{1+\frac{4 ||\Xi(m))||^2}{(1-\hetap(u))^2}}\right]}~~.
\ee
The surface of revolution degenerates to the plane perpendicular to
$(\grad\Phi)(m)$ when $\hetap(u)\rightarrow 1$. When
$\hetap(u)\rightarrow \pm \infty$, it degenerates to
the origin of $T_{m}\cM$. These remarks show that the constant roll
approximation $\hetap\approx C$ with $C$ a constant means that the
velocity $\dot{\varphi}(t)$ of the cosmological curve remains close to
the surface of revolution determined by $C$ inside $T_{\varphi(t)}\cM$.
In particular, the condition $\hetap(u)=0$ requires
$\cos\theta(u)<0$ i.e. $\theta\in (\frac{\pi}{2},\pi]$ and determines
the surface of revolution with equation:
\be
||u||=\sqrt{\Phi(\pi(u))\left[-1+\sqrt{1+4 ||\Xi(\pi(u))||^2 \cos^2\theta(u)}\right]}~~,
\ee
whose section with a plane inside $T_m\cM$ which contains the
vector $(\grad \Phi)(m)$ is plotted in figure
\ref{fig:RevSurface}. The second slow roll condition  $\hetap(\dot{\varphi}(t))=0$ requires that $\dot{\varphi}(t)$
lies close to this surface inside $T_{\varphi(t)}\cM$.

\begin{figure}[H] \centering
\includegraphics[width=.5\linewidth]{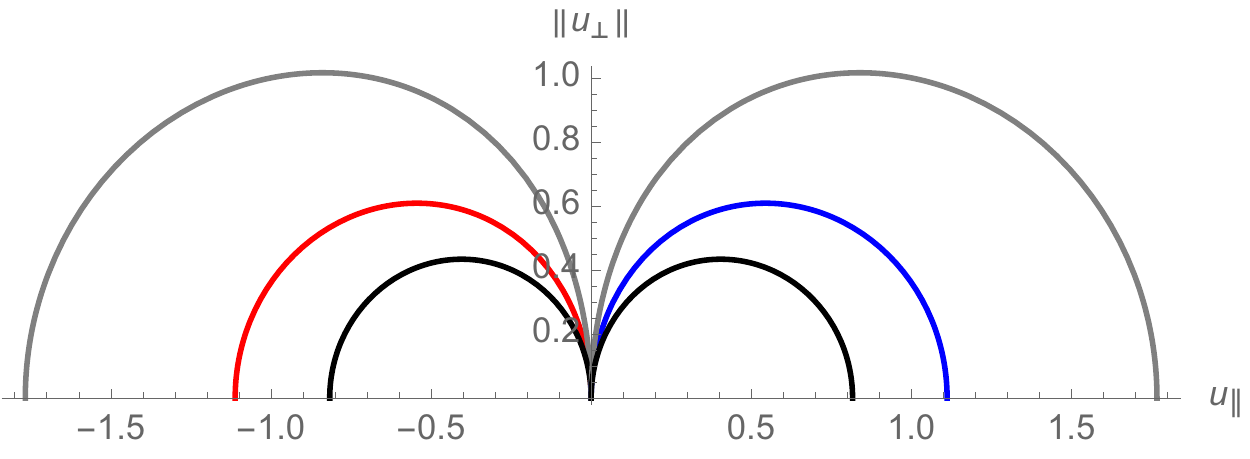}
\caption{Intersection of the surface of revolution determined inside
$T_m\cM$ by the condition $\hetap(u)=0$ for $\pi(u)=m$ with a
plane which contains the vector $(\grad \Phi)(m)$. In the figure,
we took $\Phi(m)=1$ and $||\Xi(m)||=1/2$. The horizontal
and vertical axes correspond to the projection $u_\parallel$ of $u$ on
$(\grad \Phi)(m)$ and the norm of the projection $u_\perp$ of $u$
on the plane orthogonal in $T_m\cM$ to the vector $(\grad
\Phi)(m)$.}
\label{fig:RevSurface}
\end{figure}

\subsection{The no roll condition}

\begin{definition}
The {\em no roll shell} of $(\cM,\cG,\Phi)$ is the zero level set of
the second slow roll function $\hetap$:
\be
\cF(\cM,\cG,\Phi)\eqdef \{u\in {\dot T}\cM~\vert~\hetap(u)=0\}~~,
\ee
which is a closed codimension one submanifold of ${\dot T}\cM_0$. 
\end{definition}

\noindent By the discussion of the previous subsection, the no roll
shell is a fiber sub-bundle of $T\cM_0$ whose fiber at each $m\in \cM_0$
is a connected surface of revolution around the axis determined by the vector
$(\grad \Phi)(m)\in T_m\cM$.

\begin{definition}
We say that a cosmological curve $\varphi:I\rightarrow \cM_0$
satisfies the {\em no roll condition} at time $t\in I_\reg$ if
$\dot{\varphi}(t)\in \cF(\cM,\cG,\Phi)$ i.e. if
$\hetap(\dot{\varphi}(t))=0$.
\end{definition}

\begin{remark}
The no roll condition at time $t\in I_\reg $ means that the covariant
acceleration $\nabla_t\varphi(t)$ is orthogonal on $\dot{\varphi}(t)$.
Since $\cG(\dot{\varphi},\nabla_t
\dot{\varphi})=\frac{1}{2}\frac{\dd}{\dd t}
||\dot{\varphi}||^2=||\dot{\varphi}||\frac{\dd}{\dd t}
||\dot{\varphi}||$, this amounts to the condition:
\be
\frac{\dd}{\dd t}||\dot{\varphi}(t)||=0~~,
\ee
which means that the proper length parameter $s$ along the cosmological curve
(which satisfies $\dot{s}=||\dot{\varphi}(t)||$) has vanishing second derivative at
time $t$:
\be
\ddot{s}=0~~.
\ee
\end{remark}

\subsection{The second slow roll condition}

\begin{definition}
Let $\epsilon\in (0,1]$. We say that a vector $u\in {\dot T}\cM$
satisfies the {\em second slow roll condition} of $(\cM,\cG,\Phi)$
with parameter $\epsilon>0$ if $|\hetap(u)|<\epsilon$.
The {\em second slow roll region} of $(\cM,\cG,\Phi)$ at parameter $\epsilon$
is the open subset of ${\dot T}\cM$ defined through:
\be
\cS^2_\epsilon(\cM,\cG,\Phi)\eqdef \{u\in {\dot T}\cM~\vert~|\hetap(u)|<\epsilon\}\subset {\dot T}\cM~~.
\ee
\end{definition}

\noindent Notice that $\cF(\cM,\cG,\Phi)\subset
\cS^2_\epsilon(\cM,\cG,\Phi)$ for all $\epsilon\in (0,1]$.

\begin{definition}
We say that a cosmological curve $\varphi:I\rightarrow \cM$ satisfies
the {\em second slow roll condition} with parameter $\epsilon$ at $t\in
I_\reg$ if $\dot{\varphi}(t)\in \cS^2_\epsilon(\cM,\cG,\Phi)$, i.e. if
the vector ${\dot \varphi}(t)$ satisfies the second slow roll
condition of $(\cM,\cG,\Phi)$ with parameter $\epsilon$.
\end{definition}

\noindent Since $\cS^2_\epsilon(\cM,\cG,\Phi)$ is an open subset of
${\dot T}\cM$ and the canonical lift
$\dot{\varphi}_\reg:I_\reg\rightarrow {\dot T}\cM$ of
$\varphi_\reg\eqdef \varphi\vert_{I_\reg}$ is continuous, the set:
\be
I_\epsilon^2\eqdef \dot{\varphi}_\reg^{-1}(\cS^1_\epsilon(\cM,\cG,\Phi))\subset I_\reg
\ee
is an open subset of $I_\reg$. 

\paragraph{Vectors of fixed norm which satisfy the second slow roll condition.}

When $u\in {\dot T}_m\cM_0$, relation \eqref{etaomega} gives:
\ben
\label{etaomegam}
\hetap_m(u)=1+ \frac{\Xi_m(T(u))}{[\kappa_m(||u||)(1+\kappa_m(||u||))]^{1/2}}=1+\frac{\cos\theta(u)}{c_m(||u||)}~~.
\een
For any $m\in \cM$ and any $N>0$, we denote by $\rS_m(x)$ the sphere of
radius $x$ in the Euclidean vector space $(T_m\cM,\cG_m)$. Moreover,
we set $\rS_m\eqdef \rS_m(1)$.

\begin{prop}
\label{prop:eta}
Suppose that $m\in\cM_0$ is not a critical point of $\Phi$. Then the
image of the sphere $\rS_m(x)\subset {\dot T}_m\cM$ of radius $N>0$
through the map $\hetap_m:\dot{T}_m\cM\rightarrow \R$ is given by:
\ben
\label{etamimage}
\hetap_m(\rS_m(x))=\left[1-\frac{1}{c_m(x)},1+\frac{1}{c_m(x)}\right]~~.
\een
Thus for any $N>0$ and any $\epsilon\in
\left[1-\frac{1}{c_m(x)},1+\frac{1}{c_m(x)}\right]$ there exists a vector $u\in {\dot
T}_m\cM$ which satisfies the conditions:
\ben
\label{etaeq}
\hetap(u)=\epsilon~~\mathrm{and}~~||u||=N~~.
\een
In particular, for any $N>0$ which satisfies the inequality: 
\ben
\label{Lineq}
c_m(x)\leq 1~~,
\een
there exists a vector $u\in {\dot T}_m\cM$ such that:
\ben
\label{etavanish}
\hetap_m(u)=0~~\mathrm{and}~~||u||=N~~.
\een
\end{prop}

\begin{proof}
Relation \eqref{etaomegam} implies that the image of $\rS_m$ through
$\hetap_m$ is given by \eqref{etamimage}.  The remaining statements
follows immediately from \eqref{etamimage}.
\end{proof}

\begin{remark}
Relation \eqref{Lineq} is the defining equation of the conservative closure $\overline{C_1}$ and hence is equivalent with
(see equation \eqref{ucons3}):
\ben
\label{Lineq2}
\kappa_m(x)\leq \frac{1}{2}[-1+\sqrt{1+4||\Xi_m||^2}]\Longleftrightarrow N^2\leq -\Phi(m)+\sqrt{\Phi(m)^2+M_0^2 ||(\dd \Phi)(m)||^2}~~.
\een
\end{remark}

\noindent Recall that $A_1$ denotes the conservative bound function at parameter $1$ (see \eqref{Adef}):
\be
A_1\eqdef \sqrt{ -\Phi+\sqrt{\Phi^2+M_0^2 ||(\dd \Phi)||^2}}~~.
\ee

\begin{cor}
\label{cor:etasol}
For any non-critical point $m\in \cM_0$ and any real number
\be
N\in (0,A_1(m)]~~,
\ee
the equation ${\hat \eta}_m^\parallel(u)=0$ has a solution $u\in {\dot
T}_m\cM$ of norm $||u||=N$. The set $\cF(\cM,\cG,\Phi)\cap \rS_m(x)$
of such solutions is a sphere of dimension $d-2$ when $N\in (0,A_1(m))$
and is reduced to a point when $N=A(m)$. When $N>A(m)$, the set
$\cF(\cM,\cG,\Phi)\cap \rS_m(x)$ is empty.
\end{cor}

\begin{proof}
Follows immediately from Prop \ref{prop:eta}.
\end{proof}

\noindent The corollary implies that the no roll shell intersects the conservative
region for any parameter $\epsilon\in (0,1]$:
\be
\cF(\cM,\cG,\Phi)\cap C_\epsilon(\cM,\cG,\Phi)\neq 0~~\forall \epsilon\in (0,1]~~.
\ee

\subsection{The second order slow roll conditions}

\begin{definition}
Let $\epsilon_1,\epsilon_2\in (0,1]$. We say that a vector $u\in {\dot
T}\cM$ satisfies the {\em second order slow roll conditions} with
parameters $\epsilon_1$, $\epsilon_2$ if:
\be
\kappa(u)<\epsilon_1~~\mathrm{and}~~|\hetap(u)|<\epsilon_2~~.
\ee
The {\em second order slow roll region} of $(\cM,\cG,\Phi)$ at parameters
$\epsilon_1$ and $\epsilon_2$ is the open subset of ${\dot T}\cM$
defined through:
\be
\cS_{\epsilon_1,\epsilon_2}(\cM,\cG,\Phi)\eqdef \cS^1_{\epsilon_1}(\cM,\cG,\Phi)\cap  \cS^2_{\epsilon_2}(\cM,\cG,\Phi)\subset \dot{T}\cM~~.
\ee
\end{definition}

\noindent Corollary \ref{cor:etasol} and relation \eqref{picS1} imply:

\begin{prop}
For any $\epsilon_1,\epsilon_2\in (0,1]$, we have:
\be
\pi(\cS_{\epsilon_1,\epsilon_2}(\cM,\cG,\Phi))=\cM_0~~ \mathrm{and}~~\pi(\cS_{\epsilon_1}^1(\cM,\cG,\Phi)\cap \cF(\cM,\cG,\Phi))=\cM_0
\ee
\end{prop}

\begin{definition}
We say that a cosmological curve $\varphi:I\rightarrow \cM$ satisfies
the {\em second order slow roll condition} with parameters
$\epsilon_1,\epsilon_2\in (0,1]$ at $t\in I_\reg$ if
$\dot{\varphi}(t)\in \cS_{\epsilon_1,\epsilon_2}(\cM,\cG,\Phi)$.
\end{definition}

\noindent  Since $\cS_{\epsilon_1,\epsilon_2}(\cM,\cG,\Phi)$ is an
open subset of ${\dot T}\cM$ and the canonical lift
$\dot{\varphi}_\reg:I_\reg\rightarrow {\dot T}\cM$ of
$\varphi_\reg:I_\reg\rightarrow \cM$ is continuous, the set:
\be
I_{\epsilon_1,\epsilon_2}\eqdef \dot{\varphi}_\reg^{-1}(\cS_{\epsilon_1,\epsilon_2}(\cM,\cG,\Phi))=I^1_{\epsilon_1}\cap I^2_{\epsilon_2}\subset I_\reg
\ee
is an open subset of $I_\reg$.

\subsection{The slow roll-fast turn (a.k.a. slow roll conservative) conditions}

\noindent Since:
\be
{\hat \veta}(u)=\hetap(u)n(u)+{\hat \veta}^\perp(u)~~\forall u\in {\dot T}\cM~~,
\ee
the conditions:
\ben
\label{vetaconds}
|\hetap(u)|\ll 1~~,~~||{\hat \veta}^\perp(u)||\gg 1
\een
are equivalent with:
\be
|\hetap(u)|\ll 1~~\mathrm{and}~~||{\hat \veta}(u)||\gg 1~~.
\ee
The second of these amounts to $||\q(u)||\gg 1$ i.e. $c(u)\ll 1$, where:
\be
c(u)\eqdef \frac{1}{||\q(u)||}= \frac{1}{M_0}\frac{||u||\left[||u||^2+2\Phi(\pi(u))\right]^{1/2}}{||(\dd \Phi)(\pi(u))||}
\ee
is the {\em conservative function} of Section \ref{sec:cons}. Hence
conditions \eqref{vetaconds} are equivalent with:
\be
c(u)\ll 1~~\mathrm{and}~~|\hetap(u)|\ll 1~~.
\ee
In particular, the defining conditions of the {\em slow roll-fast turn approximation}:
\be
\bepsilon(u)\ll 1~~,~~|\hetap(u)|\ll 1~~,~~||{\hat \veta}^\perp(u)||\gg 1~~.
\ee
are equivalent with the {\em slow roll conservative conditions}:
\be
c(u)\ll 1~~,~~\kappa(u)\ll 1~~,~~~|\hetap(u)|\ll 1~~.
\ee

\begin{definition}
Let $\epsilon_1,\epsilon_2,\epsilon_3\in (0,1]$. The {\em slow roll
conservative region} of $(\cM,\cG,\Phi)$ is the open subset of ${\dot
T}\cM_0$ defined through:
\be
C_{\epsilon_1,\epsilon_2,\epsilon_3}(\cM,\cG,\Phi)\eqdef \cS_{\epsilon_1,\epsilon_2}(\cM,\cG,\Phi)\cap C_{\epsilon_3}(\cM,\cG,\Phi)\subset {\dot T}\cM_0~~. 
\ee
Thus $C_{\epsilon_1,\epsilon_2,\epsilon_3}(\cM,\cG,\Phi)$ consists of
those non-zero tangent vectors $u\in \dot{T}\cM$ which satisfy the
conditions:
\be
\kappa(u)<\epsilon_1~~,~~\hetap(u)<\epsilon_2~~,~~c(u)<\epsilon_3~~.
\ee
\end{definition}

\noindent Equation \eqref{ckappaomega} gives:
\ben
\label{normomega}
||\Xi(\pi(u))||=\frac{[\kappa(u)(1+\kappa(u))]^{1/2}}{c(u)}~~.
\een
Thus \eqref{omegaeta} takes the form:
\ben
\label{hatomega}
{\hat \Xi}(\pi(u))(n(u))=c(u)\left[-1+\hetap(u)\right]~~,
\een
where:
\be
{\hat \Xi}\eqdef \frac{\Xi}{||\Xi||}
\ee
is the normalization of the one-form \eqref{omega}.

\begin{prop}
\label{prop:srft}
For any $m\in \cM$ which is not a critical point of $\Phi$ and every
$N>0$ such that:
\ben
\label{cmcond}
c_m(x)\leq 1 \Longleftrightarrow N^2\leq -\Phi(m)+\sqrt{\Phi(m)^2+M_0^2 ||(\dd \Phi)(m)||^2}
\een
there exists a non-zero vector $u\in {\dot T}_m\cM$ which satisfies
the conditions:
\be
\hetap_m(u)=0~~\mathrm{and}~~||u||=N~~.
\ee
In particular, for any $\epsilon_1,\epsilon_3\in (0,1]$
there exists a non-zero vector $u\in {\dot T}_m\cM$ which satisfies:
\be
\kappa_m(u)<\epsilon_1~~,~~\hetap_m(u)=0~~,~~c_m(u)<\epsilon_3~~.
\ee
\end{prop}

\begin{proof}
The first statement follows from Proposition \ref{prop:eta} upon
noticing that \eqref{cmcond} coincides with condition
\eqref{Lineq2}. The second statement follows from the same proposition
upon noticing that the inequalities:
\be
\kappa_m(x)< \epsilon_1~~,~~c_m(x)< \epsilon_3
\ee
amount to:
\be
N^2<\min\left(2\Phi(m)\epsilon_1,-\Phi(m)+\sqrt{\Phi(m)^2+\epsilon_3^2 M_0^2 ||(\dd \Phi)(\pi(u))||^2}\right)~~,
\ee
where we used the fact that the conservative condition $c_m(x)<\epsilon_3$ is
equivalent with \eqref{ucons2}.
\end{proof}

\begin{cor}
For $\epsilon_1,\epsilon_2,\epsilon_3\in (0,1]$, we have:
\be
\pi(C_{\epsilon_1,\epsilon_2,\epsilon_3}(\cM,\cG,\Phi))=\cM_0~~.
\ee
\end{cor}

\section{The first and second dynamical slow roll approximations}
\label{sec:SR}

The first and second slow roll conditions can be used to define {\em
dynamical approximations} of a cosmological curve by a solution curve
of an approximating equation obtained by neglecting the
corresponding parameters in the cosmological equation.

\subsection{The first dynamical slow roll approximation}

\noindent Relation \eqref{cHkappa} allows us to write the cosmological
equation \eqref{eomsingle} as:
\ben
\label{eomkappa}
\nabla_t \dot{\varphi}(t)+\frac{1}{M_0}\sqrt{2\Phi(\varphi(t))} \left[1+\kappa_{\varphi}(t)\right]^{1/2}\dot{\varphi}(t)+ (\grad_{\cG} \Phi)(\varphi(t))=0~~.
\een
The {\em first dynamical slow roll approximation} consists
of neglecting $\kappa_\varphi$ in this equation, i.e.
approximating a cosmological curve $\varphi$ by the solution
$\varphi_s$ of the {\em slow cosmological equation}:
\ben
\label{eomslow}
\nabla_t \dot{\varphi}_s(t)+\frac{1}{M_0}\sqrt{2\Phi(\varphi_s(t))} \dot{\varphi}_s(t)+ (\grad_{\cG} \Phi)(\varphi_s(t))=0
\een
which satisfies the initial conditions:
\be
\varphi_s(0)=\varphi(0)~~\mathrm{and}~~\dot{\varphi}_s(0)=\dot{\varphi}(0)~~.
\ee
The approximation is accurate around $t=0$ provided that
$\kappa_\varphi(0)\ll 1$, which amounts to the condition:
\be
\kappa(\varphi(0))=\kappa(\varphi_s(0))=\frac{||\dot{\varphi}(0)||^2}{2\Phi(\varphi(0))}\ll 1~~.
\ee
A necessary condition for the dynamical first slow
roll approximation $\varphi(t)\approx \varphi_s(t)$ to remain accurate
at $t\neq 0$ is that we have:
\ben
\label{kappas}
\kappa(\dot{\varphi}_s(t))\ll 1\Longleftrightarrow ||\dot{\varphi}(t)||\ll \sqrt{2\Phi(\varphi(t))}~~,
\een
where:
\be
\kappa(\dot{\varphi}_s(t))=\frac{||\dot{\varphi}_s(t)||^2}{2\Phi(\varphi_s(t))}~~.
\ee
Equation \eqref{kappas} requires that the lift of $\varphi$ to $T\cM$
be contained within a tubular neighborhood of the zero
section whose radius is much smaller that $\sqrt{2\Phi}$.
This condition constraints the time interval around zero for which the
approximation can remain accurate.

The geometric equation \eqref{eomslow} defines the {\em slow flow} on
the tangent bundle $T\cM$.  Notice that the stationary points of this
flow coincide with those of the cosmological flow, being given by the
trivial lifts of the critical points of $\Phi$. The cosmological
energy:
\be
E_{\varphi_s}(t)\eqdef E(\dot{\varphi}_s(t))=\frac{1}{2}||\dot{\varphi}_s(t)||^2+\Phi(\varphi_s(t))
\ee
of $\varphi_s$ satisfies:
\be
\frac{\dd E_{\varphi_s}(t)}{\dd t}=-\frac{\sqrt{2\Phi(\varphi_s(t))}}{M_0}||\dot{\varphi}_s(t)||^2
\ee
and attains its minima at the stationary points of the slow flow. It
follows that $E:T\cM\rightarrow \R_{>0}$ is a Lyapunov function for
the geometric dynamical system defined by \eqref{eomslow} on
$T\cM$. When $\cM$ is compact, the usual argument shows that the slow
flow is future complete and that the $\omega$-limit point of any slow solution
curve is a critical point of $\Phi$.

\subsection{The second dynamical slow roll approximation}

\noindent Relation \eqref{nablaheta} allows us to write the cosmological equation as:
\be
\cH_\varphi(t)(1-\hetap_\varphi(t))\dot{\varphi}(t)-||\dot{\varphi}(t)||\heta^\perp_\varphi(t)+(\grad\Phi)(\varphi(t))=0~~.
\ee
The {\em second dynamical slow roll approximation} consists of
neglecting $\hetap_\varphi(t)$ in this equation, i.e. neglecting the
parallel component of the covariant acceleration:
\be
[\nabla_t\dot{\varphi}(t)]^\parallel=\cG(\nabla_t\dot{\varphi}(t),T_\varphi(t))=\frac{\dd}{\dd t} ||\dot{\varphi}(t)||
\ee
in the cosmological equation. This amounts to approximating a cosmological curve $\varphi$ by the
solution $\varphi_\sigma$ of the {\em no second roll equation}:
\ben
\label{eom2sr}
[\nabla_t\dot{\varphi}_\sigma(t)]^\perp+\cH_{\varphi_\sigma}(t)\dot{\varphi}_\sigma(t)+(\grad\Phi)(\varphi_\sigma(t))=0~~
\een
which satisfies the initial conditions:
\be
\varphi_\sigma(0)=\varphi(0)~~\mathrm{and}~~\varphi_\sigma(t)=\varphi(t)~~.
\ee
Recall that (see \eqref{compacc}):
\ben
\label{nablatperp}
[\nabla_t\dot{\varphi}_\sigma(t)]^\perp=\nabla_t\dot{\varphi}_\sigma(t)-\left(\frac{\dd}{\dd t}\log ||\dot{\varphi}_\sigma(t)||\right)\dot{\varphi}_\sigma(t)~~.
\een
When $\dot{\varphi}(t)=0$, we define:
\be
[\nabla_t\dot{\varphi}(t)]^\perp\eqdef \nabla_t\dot{\varphi}(t)~~(t\in I_\sing)~~,
\ee
which agrees with the limit of \eqref{nablatperp} when $t$ approaches a singular point of $\varphi$.

The approximation is accurate at $t=0$ provided that:
\be
|\hetap_\varphi(0)|=|\hetap(\dot{\varphi}(0))|=|\hetap(\dot{\varphi}_\sigma(0))|=\Big{|} 1+\frac{\cos\theta(\dot{\varphi}(0))}{c(\dot{\varphi}(0))}\Big{|}  \ll 1~~.
\ee
A necessary condition that the approximation remains accurate at $t\neq 0$ is that we have:
\be
|\hetap(\dot{\varphi}_\sigma(t))|=\Big{|}1+\frac{\cos\theta(\dot{\varphi}_\sigma(t))}{c(\dot{\varphi}_\sigma(t))}\Big{|}\ll 1~~,
\ee
a condition which constrains the time interval around the origin on which the approximation can be applied.

Projecting \eqref{eom2sr} on the direction of
$T_{\varphi_\sigma}(t)$ and on the hyperplane inside
$T_{\varphi_\sigma(t)}\cM$ which is orthogonal to
$T_{\varphi_\sigma}(t)$ gives:
\beqan
\label{2sreqs}
&& \cH_{\varphi_\sigma}(t)||\dot{\varphi}_\sigma(t)||=-||(\dd\Phi)(\varphi_\sigma(t))||\cos\theta_{\varphi_\sigma}(t) \nn\\
&& ||\dot{\varphi}_\sigma(t)||^2 \chi_{\varphi_\sigma}(t)n_{\varphi_\sigma}(t)=(\grad\Phi)^\perp(\varphi_\sigma(t))~~,
\eeqan
where noticed that:
\be
[\nabla_t\dot{\varphi}_\sigma(t)]^\perp=||\dot{\varphi}_\sigma(t)||\nabla_tT_{\varphi_c}(t)=||\dot{\varphi}_\sigma(t)||^2 \chi_{\varphi_\sigma}(t)n_{\varphi_\sigma}(t)~~.
\ee
The second condition in \eqref{2sreqs} means that the oscullating
plane of $\varphi_\sigma$ contains the vector
$(\grad\Phi)(\varphi_\sigma(t))$ and that we have:
\ben
\label{2sreq2}
||\dot{\varphi}_\sigma(t)||^2 \chi_{\varphi_\sigma}(t)=||(\dd\Phi)(\varphi_\sigma(t))||\sin\theta_{\varphi_\sigma}(t)~~.
\een
This relation and the first equation in \eqref{2sreqs} determines the
principal curvature of $\varphi_\sigma$ in terms of its speed as:
\be
\chi_{\varphi_\sigma}(t)=\frac{||(\dd\Phi)(\varphi_\sigma(t))||}{||\dot{\varphi}_\sigma(t)||^2 }\left[1-\left(\frac{\cH_{\varphi_\sigma}(t)||\dot{\varphi}_\sigma(t)||}{||(\dd\Phi)(\varphi_\sigma(t))||}\right)^2\right]^{1/2}~~,
\ee
i.e.:
\ben
\chi_{\varphi_\sigma}(t)=\frac{||\Xi(\varphi_\sigma(t))||}{M_0\kappa(\dot{\varphi}_\sigma(t))}
\sqrt{1-c(\dot{\varphi}_\sigma(t))^2}~~.
\een
Notice that the first equation in \eqref{2sreqs} amounts to:
\be
c(\dot{\varphi}_\sigma(t))=-\cos\theta(\dot{\varphi_\sigma}(t))~~,
\ee
which shows that the flow defined by \eqref{eom2sr} on $T\cM$ preserves the no roll shell $\cF(\cM,\cG,\Phi)$. 

\section{The conservative and dissipative approximations}
\label{sec:CD}

In this section we discuss two dynamical approximations controlled by
the conservative parameter $c_\varphi$, namely the conservative and
dissipative approximations, which are obtained by taking this
parameter to be very small or very large. We start by explaining the
relation of these approximations to approximations controlled by the
norm of $\veta_\varphi$.

\subsection{Dynamical approximations controlled by the norm of the relative acceleration}

\noindent It is natural to consider the approximants of the
cosmological equation which are obtained by requiring that
$||\hveta_\varphi(t)||$ is very small, very large or close to one.

\paragraph{The gradient flow approximation.}

The {\em gradient flow
approximation} of \cite{genalpha} consists of replacing the cosmological
equation \eqref{eomsingle} with the {\em modified gradient flow
equation}:
\be
\cH_\varphi(t)\dot{\varphi}(t)+(\grad\Phi)(\varphi(t))=0~~,
\ee
whose integral curves are obtained from the gradient flow curves of
$\Phi$ by a certain reparameterization. This approximation (which was
further discussed in \cite{modular}) is accurate when the {\em small
relative acceleration condition} $||\veta_{\varphi}(t)||\ll 1$ holds.
Notice that this condition implies the small longitudinal relative
acceleration condition $|\hetap(u)|\ll 1$ and hence the gradient flow
approximation implies the second slow roll approximation.

Combining the gradient flow approximation with the slow motion
approximation (i.e. further neglecting $\kappa_\varphi(t)$ in the
rescaled Hubble parameter
$\cH_\varphi(t)=\frac{\sqrt{2\Phi(t)}}{M_0}(1+\kappa_\varphi(t))^{1/2}$)
results in the {\em IR approximation} of \cite{ren}. The latter
replaces \eqref{eomsingle} with the equation:
\be
\frac{1}{M_0}\sqrt{2\Phi(\varphi(t))}\dot{\varphi}(t)+(\grad\Phi)(\varphi(t))=0~~,
\ee
which is equivalent with the gradient flow equation of the {\em
classical effective potential} $V\eqdef M_0 \sqrt{2\Phi}$. 
The IR approximation is accurate when $\kappa_\varphi(t)\ll 1$ and
$||\veta_\varphi(t)||\ll 1$ and is a natural generalization to
multifield models of the second order slow roll approximation of
one-field cosmological models.

Notice that relation \eqref{veta} gives:
\ben
\label{hvetanorm}
||\hveta(u)||=\sqrt{1+\frac{1}{c(u)^2}+\frac{2\cos\theta(u)}{c(u)}}\in \left[\big{\vert} 1-\frac{1}{c(u)}\big{\vert} ,1+\frac{1}{c(u)}\right]~~.
\een
Thus:
\begin{itemize}
\item The conservative condition $c(u)\ll 1$ is {\em
equivalent} with the {\em large relative acceleration condition}
$||\veta(u)||\gg 1$, namely it forces $||\veta(u)||\approx
\frac{1}{c(u)}$.
\item The dissipative condition $c(u)\gg
1$ is {\em equivalent} with the {\em unit relative acceleration condition}
$||\veta(u)||\approx 1$.
\end{itemize}
On the other hand, the {\em small relative acceleration condition}
$||\veta(u)||\ll 1$ (which is used to define the gradient flow
approximation of \cite{genalpha}) does not constrain $c(u)$. Thus one
has the following approximations which are controlled by $||\hveta||$:

\begin{itemize}
\item The gradient flow approximation, which is accurate when
  $||\hveta_\varphi(t)||\ll 1$. This approximation was discussed in
  \cite{genalpha} and \cite{modular}.
\item The conservative approximation, which is accurate when
$||\hveta_\varphi(t)||\gg 1$, i.e. when $c_\varphi(t)\ll 1$.
\item The dissipative approximation, which is accurate when
$||\hveta_\varphi(t)||\approx 1$, i.e. when $c_\varphi(t)\gg 1$.
\end{itemize}

\noindent Each of these can be combined with the slow motion
approximation, which is accurate when $\kappa_\varphi(t)\ll 1$, i.e.
when $\bepsilon_\varphi(t)\ll 1$. We refer the reader to
\cite{genalpha} and \cite{modular} for further details of the gradient
flow approximation and to \cite{ren} and \cite{grad} for further
information on the IR approximation. Below, we discuss
the conservative and dissipative approximations.

\subsection{The conservative approximation}
\label{sec:cons}

\noindent The {\em conservative approximation} consists of considering only
non-critical cosmological curves and neglecting the friction term in
the cosmological equation. This approximation is accurate for a
noncritical cosmological curve $\varphi:I\rightarrow \cM_0$ when the
{\em kinematic conservative condition}:
\ben
\label{ccond}
c_\varphi(t) \ll 1 \Longleftrightarrow \dot{\varphi}(t)\in C_\epsilon(\cM,\cG,\Phi)~~\mathrm{for~some~positive}~\epsilon \ll 1
\een
is satisfied. Suppose that $0\in I$. When \eqref{ccond} holds at
$t=0$, the noncritical cosmological curve $\varphi$ is
well-approximated for small $|t|$ by the solution
$\varphi_c:I_c\rightarrow \cM$ of the {\em conservative equation} of the
scalar triple $(\cM,\cG,\Phi)$:
\ben
\label{cons}
\nabla_t \dot{\varphi}_c(t)+ (\grad\Phi)(\varphi_c(t))=0
\een
which satisfies the initial conditions:
\ben
\label{inconscond}
\varphi_c(0)=\varphi(0)~~\mathrm{and}~~\dot{\varphi}_c(0)=\dot{\varphi}(0)~~.
\een
In this case the approximation remains accurate for those cosmological times
$t$ close to zero which satisfy $\dot{\varphi}(t)\in
C_\epsilon(\cM,\cG,\Phi)$. This condition determines a relatively
open subset of $I$ whose connected component which contains $0$ is the
time interval on which the approximation remains accurate.

\paragraph{Conservation of energy for the conservative approximant.}

Notice that \eqref{cons} is the equation of the motion of a particle of
unit mass in the Riemannian manifold $(\cM,\cG)$ in the presence of
the potential $\Phi$. This motion is conservative in the sense that
the energy:
\ben
\label{Evarphi0}
E_{\varphi_c}\eqdef \frac{1}{2}||\dot{\varphi}_c(t)||^2+\Phi(\varphi_c(t))
\een
is independent of $t$ when $\varphi_c$ is a solution of \eqref{cons}. The
initial conditions \eqref{inconscond} determine the energy as:
\ben
\label{Econs}
E_{\varphi_c}=\frac{1}{2}||\dot{\varphi}(0)||^2+\Phi(\varphi(0))=E_\varphi(0):=E_0~~,
\een
where:
\be
E_{\varphi}(t)\eqdef \frac{1}{2}||\dot{\varphi}(t)||^2+\Phi(\varphi(t))
\ee
is the cosmological energy of $\varphi$, which is a
strictly-decreasing function of $t$ when $\varphi$ is not constant. We
have (see eq. \eqref{cME} for the definition of the set $\cM(E_0)$):
\be
\Phi(\varphi_c(t))\leq E_0~~\mathrm{i.e.}~~\varphi_c(t)\subset \cM(E_0)
\ee
since $||\dot{\varphi}_c(t)||^2\geq 0$. Equations \eqref{Evarphi0} and \eqref{Econs} give:
\ben
\label{speednormcons}
||\dot{\varphi}_c(t)||=\sqrt{2[E_0-\Phi(\varphi_c(t))]}~~
\een
and the rescaled Hubble parameter of $\varphi_c$:
\be
\cH_{\varphi_c}(t)=\cH(\dot{\varphi}_c(t))=\frac{1}{M_0}\sqrt{2E_0}
\ee
is independent of $t$. In particular, the e-fold function of
$\varphi_c$ is given by:
\ben
\label{cNcons}
\cN_{\varphi_c}(T)\eqdef \frac{1}{3}\int_0^T \dd t \cH_{\varphi_c}(t)=\frac{T}{3M_0}\sqrt{2E_0}
\een
and hence is a linear function of $T$. Moreover, we have:
\ben
\label{kappacons}
\kappa_{\varphi_c}(t)=\frac{||\dot{\varphi}_c(t)||^2}{2\Phi(\varphi_c(t))}=\frac{E_0-\Phi(\varphi_c(t))}{\Phi(\varphi_c(t))}=\frac{E_0}{\Phi(\varphi_c(t))}-1~~
\een
and the slow roll parameter of $\varphi_c$ is given by:
\ben
\bepsilon_{\varphi_c}(t)=\frac{3\kappa_{\varphi_c}(t)}{1+\kappa_{\varphi_c}(t)}=3\left[1-\frac{\Phi(\varphi_c(t))}{E_0}\right]~~.
\een
Thus $\dot{\varphi}_c(t)$ lies in the inflation region of $(\cM,\cG,\Phi)$ iff:
\be
\kappa_{\varphi_c}(t)<\frac{1}{2}\Longleftrightarrow \Phi(\varphi_c(t))>\frac{2E_0}{3}~~,
\ee
i.e. iff $\varphi_c(t)$ lies in the following sublevel set of $\Phi$:
\ben
\label{cRcons}
\cM(2E_0/3)\eqdef \{m\in \cM~\vert~ \Phi(\varphi_c(t))<\frac{2E_0}{3}\}\subset \cM
\een
which we call the {\em conservative inflation region} of $\cM$ at energy $E_0$.

\paragraph{The potential conservative condition.}

The conservative approximation $\varphi(t)\approx \varphi_c(t)$ implies that the
conservative parameter of $\varphi$ is approximated as
$c_\varphi(t)\approx c_{\varphi_c}(t)$, where:
\ben
\label{cphic}
c_{\varphi_c}(t)\eqdef c(\dot{\varphi}_c(t))=\frac{\cH(\dot{\varphi}_c(t)) ||\dot{\varphi}_c(t)||}{||(\dd \Phi)(\varphi_c(t))||}=
\frac{2}{M_0}\frac{\left[E_0(E_0-\Phi(\varphi_c(t)))\right]^{1/2}}{||(\dd \Phi)(\varphi_c(t))||}=c_{E_0}(\varphi_c(t))~~.
\een
Here $c_{E_0}:\cM_0(E_0)\rightarrow \R$ is the {\em $E_0$-reduced
conservative function} of $(\cM,\cG,\Phi)$, which is defined through:
\be
c_{E_0}\eqdef \frac{2}{M_0}\frac{\left[E_0(E_0-\Phi)\right]^{1/2}}{||\dd \Phi||}=\frac{\left[E_0(E_0-\Phi)\right]^{1/2}}
{||\Xi|| \Phi}~~\forall E_0> 0~~.
\ee
The initial conditions \eqref{inconscond} give:
\be
c_{\varphi}(0)=c({\dot \varphi}(0))=c(\dot{\varphi}_c(0))=c_{E_0}(\varphi(0))~~.
\ee
Consistency with \eqref{ccond} requires that the {\em potential
conservative condition}:
\ben
\label{potconscond}
c_{\varphi_c}(t)\ll 1\Longleftrightarrow c_{E_0}(\varphi_c(t))\ll 1
\een
is satisfied, i.e. that there exists a positive $\epsilon'\ll 1$ such that
\ben
\label{potconscond1}
\varphi_c(t)\in C_{E_0,\epsilon'}(\cM,\cG,\Phi)\eqdef \{m\in \cM_0~\vert~c_{E_0}(m)<\epsilon' \}\subset \cM_0~~.
\een
This condition constrains the cosmological times $t\neq 0$ for which
the conservative approximation can remain accurate.

\begin{remark}
We have:
\be
\frac{\dd \Phi(\varphi_c(t))}{\dd t}=(\dd \Phi)(\dot{\varphi}_c(t))~~
\ee
and:
\be
\frac{\dd^2 \Phi(\varphi_c(t))}{\dd t^2}=\Hess(\Phi)(\dot{\varphi}_c(t),\dot{\varphi}_c(t))+(\dd \Phi)(\nabla_t\dot{\varphi}_c(t))=\Hess(\Phi)(\dot{\varphi}_c(t),\dot{\varphi}_c(t))-
||\dd \Phi||^2~~.
\ee
\end{remark}

\subsection{The dissipative approximation}
\label{sec:friction}

\noindent The {\em dissipative approximation} consists of considering only
non-critical cosmological curves and neglecting the friction term in
the cosmological equation. This approximation is accurate for a
noncritical cosmological curve $\varphi:I\rightarrow \cM_0$ when the
{\em kinematic dissipative condition}:
\ben
\label{dcond}
c_\varphi(t) \gg 1
\een
is satisfied. Suppose that $0\in I$. When \eqref{ccond} holds at
$t=0$, the noncritical cosmological curve $\varphi$ is
well-approximated for small $|t|$ by the solution
$\varphi_d:I_c\rightarrow \cM$ of the {\em dissipative equation} of the
scalar triple $(\cM,\cG,\Phi)$:
\ben
\label{dissip}
\nabla_t \dot{\varphi}_d(t)+ \cH(\varphi_d(t)) \dot{\varphi}_d(t)=0
\een
which satisfies the initial conditions:
\ben
\label{indissipcond}
\varphi_d(0)=\varphi(0)~~\mathrm{and}~~\dot{\varphi}_d(0)=\dot{\varphi}(0)~~.
\een
In this case the approximation remains accurate for those cosmological
times $t$ close to zero which satisfy $c(\dot{\varphi}(t))\gg 1$. This
condition determines a relatively open subset of $I$ whose connected
component which contains $0$ is the time interval on which the
approximation remains accurate.

The dissipative equation \eqref{dissip} is a modified
(a.k.a. reparameterized) geodesic equation.  Indeed, let $s=s(t)$ be
an increasing parameter along the curve $\varphi$,
chosen such that $s(0)=0$. We have:
\be
\dot{\varphi}_d=\dot{s}\varphi'_d~~\mathrm{and}~~ \nabla_t\dot{\varphi}_d=\ddot{s} \varphi'_d+\dot{s}^2 \nabla_s \varphi'_d~~,
\ee
where the primes denote derivatives with respect to $s$. Hence \eqref{dissip} is
equivalent with:
\be
\nabla_s \varphi'_d(s)+\frac{\ddot{s}+\dot{s}\cH(\dot{s}\varphi'_d(s))}{\dot{s}^2}=0~~
\ee
and reduces to the geodesic equation $\nabla_s\varphi'_d(s)=0$
with affine parameter $s$ provided that $s$ satisfies:
\ben
\label{scond}
\frac{\ddot{s}}{\dot{s}^2}+\frac{1}{M_0}\sqrt{||\varphi'_d(s)||^2+2\frac{\Phi(\varphi_d(s))}{\dot{s}^2}}=0~~.
\een
Since:
\be
\frac{\ddot{s}}{\dot{s}^2}=-\frac{\dd}{\dd t}\left(\frac{1}{\dot{s}}\right)=-\dot{s}t''(s)=-\frac{t''(s)}{t'(s)}~~,
\ee
condition \eqref{scond} is equivalent with the following ODE for the inverse function $t=t(s)$:
\ben
\label{tode}
t''(s)-\frac{1}{M_0}[||\varphi'_d(s)||^2+2\Phi(\varphi_d(s)) t'(s)^2]^{1/2} t'(s)=0~~.
\een
It follows that $\varphi_d$ is a reparameterized geodesic of $(\cM,\cG)$, namely $\varphi(s)=\eqdef \varphi(t(s))$ is
a normalized geodesic. The initial conditions \eqref{indissipcond} amount to:
\be
\varphi(0)=\varphi(0)~~\mathrm{and}~~\varphi'_d(0)=t'(0)\dot{\varphi}(0)~~.
\ee
Since $||\varphi'_d(0)||=1$, the second of these conditions implies:
\be
t'(0)=\frac{1}{||\dot{\varphi}(0)||}~~,
\ee
which together with $t(0)=0$ specifies the initial condition for the
desired solution $t(s)$ of \eqref{tode}. Moreover, the initial
conditions for $\varphi_d$ can be written as:
\ben
\label{phidincond}
\varphi_d(0)=\varphi(0)~~\mathrm{and}~~\varphi'_d(0)=T_\varphi(0)=\frac{\dot{\varphi}(0)}{||\dot{\varphi}(0)||}~~.
\een
When $\varphi$ is a maximal cosmological curve, the curve
$\varphi_d(s)$ can be taken to the unique maximal normalized geodesic
of $(\cM,\cG)$ which satisfies these conditions.

For the approximation of $\varphi(t)$ by $\varphi_d(t)$ to remain
accurate for $t\neq 0$, it is necessary that the {\em effective
dissipative condition}:
\be
c_{\varphi_d}(t)\ll 1
\ee
be satisfied. Since $\varphi_d(s)$ is a normalized geodesic, we have:
\be
||\dot{\varphi}_d(t)||=\dot{s}=\frac{1}{t'(s)}~~.
\ee
Hence:
\be
\kappa_{\varphi_d}(t)=\frac{1}{2t'(s)^2 \Phi(\varphi_d(s))}
\ee
and:
\be
c_{\varphi_d}(s)=\frac{\sqrt{1+2t'(s)^2 \Phi(\varphi_d(s))}}{2 t'(s)^2 \Phi(\varphi_d(s))||\Xi(\varphi_d(s))||}=\frac{\sqrt{1+2t'(s)^2 \Phi(\varphi_d(s))}}{M_0 t'(s)^2 ||(\dd \Phi)(\varphi_d(s))||}~~.
\ee
Notice that
$c_{\varphi_d}(0)=c(\dot{\varphi}_d(0))=c(\dot{\varphi}(0))=c_\varphi(0)$
since $\dot{\varphi}_d(0)=\dot{\varphi}(0)$.

As explained in Subsection \ref{subsec:etageom}, the dissipative
condition $c(\dot{\varphi}(t))\gg 1$ forces $\hetap(\dot{\varphi}(t))$
to be very close to one. Hence the generalized ``ultra slow roll'' approximation holds
in the strongly dissipative regime, where the dissipative approximation
is accurate.

\section{The limits of large and small rescaled Planck mass}
\label{sec:Planck}

\noindent Recall from \cite{ren} that universal similarities can be used to absorb all
parameters of the model into the rescaled Planck mass $M_0$. Thus it
is natural to consider the limits when $M_0$ is very large and very
small.

When $M_0\gg 1$, the friction term can be neglected and hence
cosmological curves are well-approximated by solutions of the
conservative equation \eqref{cons}.

When $M_0\ll 1$, the scale transformation with parameter
$\epsilon=M_0$ brings \eqref{eomsingle} to the form:
\ben
\label{smallM0}
M_0^2\nabla_t \frac{\dd\varphi_{M_0}(t)}{\dd t}+\left[M_0^2||\frac{\dd\varphi_{M_0}(t)}{\dd t}||^2+2\Phi(\varphi_{M_0}(t))\right]^{1/2}\frac{\dd\varphi_{M_0}(t)}{\dd t}+ (\grad_{\cG} \Phi)(\varphi_{M_0}(t))=0~~,
\een
where $\varphi_{M_0}(t)=\varphi(t/M_0)$. Hence the limit of very small
$M_0$ coincides with the infrared limit of \cite{ren} at parameter
$\epsilon=M_0$. In this limit, the rescaled equation \eqref{smallM0}
approximates as:
\be
\frac{\dd\varphi_{M_0}(t)}{\dd t}+(\grad V_1)(\varphi_{M_0}(t))\approx 0~~,
\ee
where:
\be
V_1\eqdef \sqrt{2\Phi}
\ee
coincides with the classical effective potential of \cite{ren} for $M_0=1$.
Thus $\varphi(t)$ is well-approximated by the curve $\varphi_0(t)=\varphi_1(M_0 t)$, where $\varphi_1(t)$
is the solution of the gradient flow equation of $V_1$:
\ben
\label{gf1}
\frac{\dd\varphi_1(t)}{\dd t}+(\grad V_1)(\varphi_1(t))= 0~~
\een
which satisfies the initial condition:
\be
\varphi_1(0)=\varphi(0)~~.
\ee
Notice that \eqref{gf1} is equivalent with the condition that
$\varphi_0$ satisfies the gradient flow equation:
\ben
\label{gf}
\frac{\dd\varphi_0(t)}{\dd t}+(\grad V)(\varphi_0(t))= 0~~
\een
of the classical effective scalar potential:
\be
V=M_0 V_1=M_0\sqrt{2\Phi}
\ee
introduced in \cite{ren}. It also satisfies the initial condition:
\be
\varphi_0(0)=\varphi(0)~~.
\ee
The approximation is most accurate for {\em infrared optimal curves},
which satisfy:
\be
\dot{\varphi}(0)=-M_0(\grad V_1)(\varphi(0))\Longleftrightarrow  \dot{\varphi}(0)=-\frac{M_0}{\sqrt{2\Phi}}(\grad \Phi)(\varphi(0))~~.
\ee
This amounts to the condition that $\dot{\varphi}(0)$ belongs to the
gradient flow shell of $(\cM,\cG,V)$:
\be
\dot{\varphi}(0)=-(\grad V)(\varphi(0))~~.
\ee
The approximation remains accurate for $t\neq 0$ when one can neglect the acceleration
and kinetic energy terms in \eqref{eomsingle}, which requires:
\be
\kappa_\varphi(t)\eqdef \frac{||\dot{\varphi}(t)||^2}{2\Phi(\varphi(t))}\ll 1
\ee
and:
\be
{\tilde \kappa}_\varphi(t)\eqdef \frac{||\nabla_t \dot{\varphi}(t)||}{||(\dd\Phi)(\varphi(t))||}\ll 1~~.
\ee
Notice that ${\tilde \kappa}_\varphi(t)$ coincides with the second IR parameter of \cite{ren}.
Thus:

\

\noindent {\em In leading order, the large rescaled Planck mass limit
$M_0\gg 1$ reproduces the conservative approximation, while the small
rescaled Planck mass limit $M_0\ll 1$ reproduces the IR
approximation.}

\section{Conclusions and further directions}
\label{sec:conclusions}

We gave a careful geometric construction of natural basic observables
for multifield cosmological models with arbitrary scalar manifold
topology. Some of these first order observables are obtained by
on-shell reduction of second order observables and their definition on
the tangent bundle of the scalar manifold is not obvious a priori. We
also discussed relations between these observables and the regions of
interest which they determine within the tangent bundle of the scalar
manifold. Finally, we described a system of dynamical approximations
of the cosmological equation which are defined by imposing conditions
on these basic observables -- some of which were not considered
systematically before. This defines a hierarchy of qualitatively
distinct dynamical regimes which deserve detailed study.

Our discussion of these approximations only addressed their most basic
features and much remains to be done. In particular, each dynamical
approximation considered in the paper could be studied by expansion
techniques with a view towards extracting asymptotic approximation
schemes for cosmological curves in each
dynamical regime we have identified. This appears to be quite involved
since it requires methods from the asymptotic theory of nonlinear
geometric ODEs and the control of various bounds involving the scalar potential and
its derivatives as well as the distance function of $(\cM,\cG)$. In
principle, the dynamical approximations which we identified could
also be used to devise numerical approximation methods which could
shed light on the corresponding dynamical regimes. Finally, it would
be interesting to study these regimes in detail for the class of tame
two field models, similar to the study of the IR regime performed in
\cite{grad}. We hope to report on these and related questions in the
future.

\section*{Acknowledgments}

\noindent This work was supported by national grant PN 19060101/2019-2022 and
by a Maria Zambrano Fellowship.

\end{document}